\documentclass[11pt, a4paper]{article}
\usepackage[textwidth = 7in, textheight=25cm,nohead]{geometry}
\usepackage{ulem}
\usepackage{graphicx, amssymb, amsmath, amsthm, amstext, booktabs, float, multirow, caption, subcaption, rotating,natbib, epsfig}
\usepackage{epstopdf}
\usepackage{enumitem}
\usepackage{setspace}
\usepackage{mathtools}
\usepackage{bigstrut}
\usepackage[marginpar]{todo}   

\newcommand*\diff{\mathop{}\!\mathrm{d}} 

\newcommand{\Zd}{\mathbb{Z}^d}

\newcommand{\Norm}[1]{\left\|#1\right\|}

\newcommand{\psiq}{\psi^{(q)}}
\newcommand{\psiqjz}{\psi^{(q)}_{j,z}}
\newcommand{\betaqjz}{\beta^{(q)}_{j,z}}
\newcommand{\betahatqjz}{\hat{\beta}^{(q)} _{j,z}}
\def\fs{\kern 0.33em}
\def\ppn{\vskip 6pt \noindent }
\def\R{{\mathbb{R}}}
\def\N{{\mathbb{N}}}
\def\Z{{\mathbb{Z}}}
\def\P{{\mathbb{P}}}
\def\E{{\mathbb{E}}}
\newcommand{{\Xs}}{{\cal X}}
\newcommand{{\Ys}}{{\cal Y}}
\newcommand{{\Ls}}{{\cal L}}
\newcommand{{\Ss}}{{\cal S}}
\newcommand{{\Ms}}{{\cal M}}
\newcommand{{\Gs}}{{\cal G}}
\newcommand{{\Hs}}{{\cal H}}
\newcommand{{\Ns}}{{\cal N}}
\newcommand{{\Is}}{{\cal I}}
\newcommand{{\As}}{{\cal A}}
\newcommand{{\Bs}}{{\cal B}}
\newcommand{{\Cs}}{{\cal C}}
\newcommand{{\Rs}}{{\cal R}}
\newcommand{{\Us}}{{\cal U}}
\newcommand{{\Es}}{{\cal E}}
\newcommand{{\Fs}}{{\cal F}}
\newcommand{{\pp}}{{\mathbf p}}
\newcommand{{\Ps}}{{\cal P}}
\newcommand{{\KK}}{{\mathbf K}}
\newcommand{{\HH}}{{\mathbf H}}
\newcommand{{\II}}{{\mathbf I}}
\newcommand{{\yy}}{{\mathbf y}}
\newcommand{{\ab}}{{\mathbf a}}

\newcommand{{\toL}}{{\ \overset{\mathcal{L}}{\longrightarrow}\ }}

\bibpunct{(}{)}{,}{a}{,}{,}

\numberwithin{equation}{section}
\numberwithin{figure}{section}
\numberwithin{table}{section}
\newtheorem{lemma}{Lemma}[section]

\newtheorem{theorem}{Theorem}[section]
\newtheorem{assumption}{Assumption}[section]
\newtheorem{proposition}{Proposition}[section]
\newtheorem{corollary}{Corollary}[section]

\theoremstyle{remark}

\DeclareMathOperator{\var}{\mathbb{V}ar}

\begin{document}

\title{Shape-preserving wavelet-based multivariate density estimation}
\author{\sc{Carlos Aya Moreno}  \and \sc{Gery Geenens}\thanks{Corresponding author: {\tt ggeenens@unsw.edu.au}, School of Mathematics and Statistics, UNSW Sydney, NSW 2052 (Australia), tel +61 2 938 57032, fax +61 2 9385 7123 } \and and \qquad  \sc{Spiridon Penev}}
\date{School of Mathematics and Statistics, \\ UNSW Sydney, Australia}
\maketitle
\thispagestyle{empty}

\begin{abstract}
\noindent Wavelet estimators for a probability density $f$ enjoy many good properties, however they are not `shape-preserving' in the sense that the final estimate may not be non-negative or integrate to unity. A solution to negativity issues may be to estimate first the square-root of $f$ and then square this estimate up. This paper proposes and investigates such an estimation scheme, generalising to higher dimensions some previous constructions which are valid only in one dimension. The estimation is mainly based on nearest-neighbour-balls. The theoretical properties of the proposed estimator are obtained, and it is shown to reach the optimal rate of convergence uniformly over large classes of densities under mild conditions. Simulations show that the new estimator performs as well in general as the classical wavelet estimator, while automatically producing estimates which are {\it bona fide} densities. 
\end{abstract}

\section{Introduction}\label{sec:intro}

The mathematical theory of wavelets offers a powerful tool for approximating possibly irregular functions or surfaces. It has been successfully applied in many different fields of applied mathematics and engineering, see the classical references on the topic \citep{Meyer92,Daubechies92}, or \cite{Strang89,Strang93} for shorter reviews. In statistics, it provides a convenient framework for some nonparametric problems, in particular density estimation and regression. As opposed to most of their competitors, such as kernels or splines, wavelet-based estimators provide highly adaptive estimates by exploiting the localisation properties of the wavelets. This translates into good global properties even when the estimated function presents sharp features, such as acute peaks or abrupt changes. Indeed, wavelet estimators are (near-) optimal in some sense over large classes of functions \citep{Kerkyacharian93a,DJKP93,DJKP95,DonohoJohnstone94a,DonohoJohnstone95b,DonohoJohnstone96,DonohoJohnstone98,Fan93}. \cite{Hardle98,Vidakovic99} and \cite{Nason08} give comprehensive reviews of wavelet methods applied to statistics.

\ppn For any function $\phi: \R \to \R$, define its rescaled and translated version $\phi_{j,z} = 2^{j/2} \phi(2^j x -z)$, $j \in \N$, $z \in \Z$, as is customary in the wavelet framework. Set so-called `father' $\varphi: \R \to \R$ and `mother' $\psi: \R \to \R$ wavelets, and a certain basic `resolution' level $j_0 \in \N$. Then, the sequence $\{\varphi_{j_0,z},\psi_{j,z}; j=j_0,\ldots,\infty, z \in \Z\}$ is known to form an orthonormal basis of $L_2(\R)$ associated with a certain multiresolution analysis system \citep[Chapter 2]{Meyer92}. This means that any square-integrable function $f \in L_2(\R)$ can be expanded into that wavelet basis as
\begin{equation} f(x) = \sum_{z \in \Z} \alpha^*_{j_0,z} \varphi_{j_0,z}(x) + \sum_{j=j_0}^\infty \sum_{z \in \Z} \beta^*_{j,z} \psi_{j,z}(x), \label{eqn:uniwavexp} \end{equation}
with $\forall j \in \N$ and $z \in \Z$, $\alpha^*_{j,z} = \int_\R \varphi_{j,z}(x) f(x)\diff x$ and $\beta^*_{j,z} = \int_\R \psi_{j,z}(x) f(x)\diff x $. The term $\sum_{z \in \Z} \alpha^*_{j_0,z} \varphi_{j_0,z}(x)$ is called the `trend' at level $j_0$, while, for each level $j \geq j_0$, $\sum_{z \in \Z} \beta^*_{j,z} \psi_{j,z}(x)$ is the `detail' at level $j$. A key feature of a multiresolution representation such as (\ref{eqn:uniwavexp}) is that, for any $j \geq j_0$, the trend at level $j + 1$ coincides with the trend at level $j$ supplemented with the detail at level $j$. Specifically,
\begin{equation} \sum_{z \in \Z} \alpha^*_{j+1,z} \varphi_{j+1,z}(x) = \sum_{z \in \Z} \alpha^*_{j,z} \varphi_{j,z}(x) + \sum_{z \in \Z} \beta^*_{j,z} \psi_{j,z}(x). \label{eqn:trend}\end{equation}

\ppn When $f$ in (\ref{eqn:uniwavexp}) is a probability density, noting that $\alpha^*_{j_0,z} = \E(\varphi_{j_0,z}(X))$ and $\beta^*_{j,z}  = \E(\psi_{j,z}(X))$ paves the way for their estimation, upon observing a sample from $f$, by empirical averages, say $\widehat{\alpha}^*_{j_0,z}$ and $\widehat{\beta}^*_{j,z}$. In addition, for any practical purpose the infinite expansion (\ref{eqn:uniwavexp}) needs to be truncated after a finite number of terms, say $J \geq j_0$ -- in the wavelet jargon, one says that $f$ is approximated to the resolution level $J$. So, a wavelet estimator for $f$ writes
\[\hat{f}_J(x) = \sum_{z \in \Z} \widehat{\alpha}^*_{j_0,z} \varphi_{j_0,z}(x) + \sum_{j=j_0}^J \sum_{z \in \Z} \widehat{\beta}^*_{j,z} \psi_{j,z}(x), \]
which may ultimately include some thresholding of the estimated coefficients. Note that the sums over $z$ are essentially finite if the wavelets $\varphi$ and $\psi$ have compact support, as it is usually assumed.

\ppn Extending this framework to the multivariate case is conceptually straightforward. We assume that an orthogonal wavelet basis for $L_2(\R^d)$ is available -- see \citet[Section 3.6]{Meyer92} for details about existence of such a basis. That is, there exist functions $\varphi: \R^d \to \R$ and $\psi^{(q)}: \R^d \to \R$, $q \in Q = \{1,\ldots,2^d-1\}$, such that $\{\varphi_{j_0,z},\psi^{(q)}_{j,z}; j=j_0,\ldots,\infty, z \in \Z^d, q \in Q\}$ forms an orthonormal basis of $L_2(\R^d)$, with $\varphi_{j_0,z}\left(x\right) = 2^{d\,j_0 /2}\varphi\left(2^{j_0} x - z \right)$ and $\psiqjz(x) = 2^{d\,j /2} \psiq(2^j x - z)$. The functions $\psi^{(q)}$ are typically obtained via a tensor product construction \citep[Sections 3.3-3.4]{Meyer92}. Then, any $d$-variate function $f \in L_2(\R^d)$ can be written
\begin{equation} f(x) = \sum_{z \in \Z^d} \alpha^*_{j_0,z}\varphi_{j_0,z}(x) + \sum_{j=j_0}^\infty \sum_{z \in \Z^d} \sum_{q \in Q} \beta^{*(q)}_{j,z}  \psi^{(q)}_{j,z}(x), \label{eqn:multwavexp} \end{equation}
where $\alpha^*_{j,z} = \int_{\R^d} \varphi_{j,z}(x) f(x)\diff x$ and $\beta^{*(q)}_{j,z} = \int_{\R^d} \psi_{j,z}^{(q)}(x) f(x)\diff x$. When $f$ is a density, estimation of these coefficients, and hence of $f$ itself, follows in the same way as in one dimension.

\ppn One major drawback, though, of such wavelet-based estimators is that they are in general not `shape-preserving'. When estimating a probability density $f$, that means that the resulting estimator $\hat{f}_J$ may neither be non-negative, nor integrate to one \citep{Dechevsky97,Dechevsky98}. Usually, simple rescaling solves the integrability issue, but overcoming the non-negativity issue requires caution. One way to address it is to first construct a wavelet estimator of $g \doteq \sqrt{f}$ which, when squared up, would obviously produce an estimator of $f$ automatically satisfying the non-negativity constraint. Consider the univariate case. Clearly, $g \in L_2(\R)$, as $\int_\R g^2(x)\diff x = \int_\R f(x)\diff x = 1$, hence we can write its expansion (\ref{eqn:uniwavexp}):
\[g(x) =   \sum_{z \in \Z} \alpha_{j_0,z} \varphi_{j_0,z}(x) + \sum_{j=j_0}^\infty \sum_{z \in \Z} \beta_{j,z} \psi_{j,z}(x),\]
where
\begin{multline} \alpha_{j,z} = \int_\R \varphi_{j,z}(x) g(x)\diff x = \int_\R \varphi_{j,z}(x) \sqrt{f}(x)\diff x \ \text{ and }\  \beta_{j,z} = \int_\R \psi_{j,z}(x) g(x)\diff x = \int_\R \psi_{j,z}(x) \sqrt{f}(x)\diff x. \label{eqn:defcoef} \end{multline}
Difficulty in estimating these coefficients arises as $\alpha_{j,z} =\E(\varphi_{j,z}(X)/\sqrt{f}(X))$ and $\beta_{j,z} =\E(\psi_{j,z}(X)/\sqrt{f}(X))$ can no more be estimated directly by sample averages. \cite{Pinheiro97} got around the presence of the unknown factor $1/\sqrt{f}$ in these expectations by plugging in a pilot estimator of $f$. Rather, \cite{Penev97} suggested a more elegant construction based on order statistics and spacings. Unfortunately, direct application of their idea is limited to the univariate case, as spacings are not defined in more than one dimension. Yet, the need for a multivariate extension of the `Dechevsky-Penev' construction was explicitly called for by \cite{McFadden03} in his Nobel Prize lecture.  \cite{Cosma07} and \cite{Peter08} attempted such extension but losing much of the initial flavour of the idea.

\ppn The aim of this paper is to suggest and study a wavelet estimator of $\sqrt{f}$ directly inspired by \cite{Penev97}'s construction, hence keeping its simplicity and attractiveness, but available in any dimension. It will be shown in Section \ref{subsec:motiv} that the volume of the smallest ball centred at $x$ and covering at least $k$ observations of the sample (for some $k \geq 1$), can act in some sense as a surrogate for a `multivariate spacing'. The suggested estimator will then make use of $k$-nearest neighbour ideas, as will be formally defined in Section \ref{subsec:def}. Sections \ref{sec:coeff} and \ref{sec:gJ} respectively present the asymptotic properties of the proposed estimators of the wavelet coefficients and of the density estimator as a whole. Section \ref{sec:numE} assesses the practical  performance of the estimator through a simulation study and a real data application. Section \ref{sec:ccl} concludes and offers some perspectives of future research.

\section{Definition of the estimator}

\subsection{Motivation} \label{subsec:motiv}

Let $\Xs = \{X_1,\ldots,X_n\}$ be a random sample from an unknown $d$-dimensional distribution $F$ admitting a density $f$ on $\R^d$. Denote by $X_{(k);i}$ the $k$th closest observation from $X_i$ among the other points of $\Xs$. Define $R_{(k);i} = \|X_{(k);i}-X_i \|$ the Euclidean distance between $X_i$ and $X_{(k);i}$, and 
\begin{equation} V_{(k);i} = c_0 R_{(k);i}^d \qquad \text{ where } \qquad c_0 = \frac{\pi^{d/2}}{\Gamma(d/2+1)}, \label{eqn:vol}\end{equation}
the volume of the ball of radius $R_{(k);i}$ centred at $X_i$ -- hence it is the smallest ball centred at $X_i$ containing at least $k$ other observations from $\Xs$. It is known \cite[Proposition 2]{Ranneby05} that, conditionally on $X_i$,
\begin{equation*} n V_{(1);i} \toL  \text{Exp}\left(f(X_i)\right)\qquad \text{ as }n \to \infty,  \end{equation*}
meaning that \citep[Section 10.5]{Johnson94}
\begin{equation} \sqrt{n V_{(1);i}} \toL  \text{Rayleigh}\left(\frac{1}{\sqrt{2}\sqrt{f}(X_i)}\right)\qquad \text{ as }n \to \infty. \label{eqn:convlawnV} \end{equation}
Now, consider an arbitrary square-integrable function $\phi: \R^d \to \R$, and define
\begin{equation} S_n \doteq \frac{2}{\sqrt{\pi}}\,\frac{1}{\sqrt{n}}\sum_{i=1}^n \phi(X_i) \sqrt{V_{(1);i}}. \label{eqn:S} \end{equation}
By the Law of Iterated Expectations, we have
\[ \E(S_n) = \E\left(\frac{2}{\sqrt{\pi}} \phi(X_i)\E\left(\sqrt{nV_{(1);i}} \ \big|X_i\right)\right). \]
The expectation of a Rayleigh$(\sigma)$-random variable is known to be $\sigma \sqrt{\pi/2}$. If the convergence in law (\ref{eqn:convlawnV}) implies the convergence of the moments (this is indeed the case here as will be formally derived later), then
\[\E(S_n) \to \E\left(\frac{2}{\sqrt{\pi}}\phi(X_i) \frac{\sqrt{\pi}}{2\sqrt{f}(X_i)} \right) = \int_{\R^d} \frac{\phi(x)}{\sqrt{f}(x)}\,f(x)\diff x = \int_{\R^d} \phi(x)\sqrt{f}(x)\diff x. \]
Hence, $S_n$ is an asymptotically unbiased estimator of $\int_{\R^d} \phi(x)\sqrt{f}(x)\diff x$. This fact naturally suggests estimating the wavelet coefficients (\ref{eqn:defcoef}) by statistics of type (\ref{eqn:S}), which is the idea formally investigated in this paper.

\subsection{Definition} \label{subsec:def}

Let $g = \sqrt{f}$, where $f$ is the $d$-dimensional density to estimate. As $g \in L_2(\R^d)$ always, we have, by (\ref{eqn:multwavexp}),
\begin{equation*}  g(x) = \sum_{z \in \Z^d} \alpha_{j_0,z}\varphi_{j_0,z}(x) + \sum_{j=j_0}^\infty \sum_{z \in \Z^d} \sum_{q \in Q} \beta^{(q)}_{j,z}  \psi^{(q)}_{j,z}(x), \end{equation*}
with, for all $j \in \N$, $z \in \Z^d$ and $q \in Q$,
\[\alpha_{j,z} = \int_{\R^d} \varphi_{j,z}(x) \sqrt{f}(x)\diff x \qquad \text{ and } \qquad \beta^{(q)}_{j,z} = \int_{\R^d} \psi_{j,z}^{(q)}(x) \sqrt{f}(x)\diff x. \]
The approximation of $g$ to the resolution level $J \geq j_0$ is
\begin{align} g_{J}(x) & = \sum_{z \in \Z^d} \alpha_{j_0,z}\varphi_{j_0,z}(x) + \sum_{j=j_0}^{J} \sum_{z \in \Z^d} \sum_{q \in Q} \beta^{(q)}_{j,z}  \psi^{(q)}_{j,z}(x) \label{eqn:gJwavexp} \\
& = \sum_{z \in \Z^d} \alpha_{J+1,z}\varphi_{J+1,z}(x),  \label{eqn:gJwavexp2} \end{align}
where the second equality follows by analogy with (\ref{eqn:trend}). 

\ppn Now, motivated by the observations made in Section \ref{subsec:motiv}, we define the estimators of the wavelet coefficients $\alpha_{j,z}$'s and $\beta^{(q)}_{j,z}$'s in (\ref{eqn:gJwavexp})-(\ref{eqn:gJwavexp2}) as
\begin{align}
\hat{\alpha }_{j,z} &=  \frac{\Gamma(k)}{\Gamma(k+1/2)}\,\frac{1}{\sqrt{n}}\sum _{i=1}^n \varphi _{j,z}\left(X_i\right) \sqrt{V_{(k);i}}, \qquad j \in \N;z \in \Zd \label{eqn:coefficients1} \\
\betahatqjz &=  \frac{\Gamma(k)}{\Gamma(k+1/2)} \,\frac{1}{\sqrt{n}}\sum _{i=1}^n \psiqjz\left(X_i\right) \sqrt{V_{(k);i}}, \qquad j \in \N;z \in \Zd;  q \in Q, \label{eqn:coefficients2}
\end{align}
for some integer $k \geq 1$. The coefficient $\frac{\Gamma(k)}{\Gamma(k+1/2)}$ guarantees the consistency of these estimators, as will arise from the proof of Proposition \ref{prop:coefconsist} below. Note that, for $k=1$, $\frac{\Gamma(1)}{\Gamma(3/2)} = \frac{2}{\sqrt{\pi}}$, as it was anticipated in Section \ref{subsec:motiv}. Also, in the case $d=1$, when the volume of a ball amounts to the width of an interval, (\ref{eqn:coefficients1}) and (\ref{eqn:coefficients2}) can easily be compared to \cite{Penev97}'s estimators (their equations (3.2) and (3.3)). Although not identical, they definitely have the same flavour and are asymptotically equivalent.

\ppn Plugging (\ref{eqn:coefficients1}) and (\ref{eqn:coefficients2}) into the expansion (\ref{eqn:gJwavexp}) produces the estimator
\begin{equation}
\hat{g}_{J}(x)=\sum _{z \in \Zd} \hat{\alpha }_{j_0,z} \varphi _{j_0,z}(x) + \sum_{j=j_0}^{J} \sum_{z \in \Z^d} \sum_{q \in Q} \betahatqjz \, \psiqjz(x), \label{eqn:hatgJ}
\end{equation}
which is also
\begin{equation} \hat{g}_{J}(x) = \sum_{z \in \Z^d} \hat{\alpha}_{J+1,z}\varphi_{J+1,z}(x) \label{eqn:hatgJ2} \end{equation}
by (\ref{eqn:gJwavexp2}) and the properties of multiresolution analysis. Squaring this up provides an estimator $\hat{f}_J$ of $f$. As already noted in \cite{Penev97}, estimating $f$ by squaring up an estimate of $\sqrt{f}$ has the additional advantage of providing an easy way for normalising the density estimate. Specifically, enforcing the condition $1 = \int_{\R^d} \hat{f}(x)\diff x = \int_{\R^d} \hat{g}^2_J(x)\diff x$ amounts to imposing
\begin{equation} \sum _{z \in \Zd} \hat{\alpha}^2_{j_0,z}  + \sum_{j=j_0}^{J} \sum_{z \in \Z^d} \sum_{q \in Q} \hat{\beta}_{j,z}^{(q)2} =1, \label{eqn:unitint} \end{equation}
given that the wavelets are orthonormal. If this sum is not 1 after raw estimation of the coefficients by (\ref{eqn:coefficients1}) and (\ref{eqn:coefficients2}) but, say, another constant $\kappa$, it is enough to divide each estimated coefficient by $\sqrt{\kappa}$ for enforcing (\ref{eqn:unitint}). Conventional wavelet estimators do not enjoy such a convenient way of normalising.

\ppn In the following section, the asymptotic properties of the coefficient estimators (\ref{eqn:coefficients1}) and (\ref{eqn:coefficients2}) are obtained. The asymptotic properties of the estimator (\ref{eqn:hatgJ})-(\ref{eqn:hatgJ2}) for $\sqrt{f}$ and the ensuing estimator $\hat{f}_J = \hat{g}_J^2$ for $f$ will be obtained in Section \ref{sec:gJ}.

\section{Asymptotic properties of the estimators of the wavelet coefficients} \label{sec:coeff}

Throughout the paper we work under the following two standard assumptions.
\begin{assumption} \label{ass:sample} The sample $\Xs=\{X_1,\ldots,X_n\}$ consists of i.i.d.\ replications of a random variable $X \in \R^d$ whose distribution $F$ admits a density $f$. \end{assumption}
\begin{assumption} \label{ass:wav} The functions $\varphi$ and $\psi^{(q)}$ ($q \in Q$), have compact support on $\R^d$ and are bounded. Defining $\varphi_{j_0,z}\left(x\right) = 2^{d\,j_0 /2}\varphi\left(2^{j_0} x - z \right)$ and $\psiqjz(x) = 2^{d\,j /2} \psiq(2^j x - z)$, $\{\varphi_{j_0,z},\psi^{(q)}_{j,z}; j=j_0,\ldots,\infty, z \in \Z^d, q \in Q\}$ is an orthonormal basis of $L_2(\R^d)$.
\end{assumption}
Now, the main ingredients in (\ref{eqn:coefficients1}) and (\ref{eqn:coefficients2}) are the $V_{(k);i}$'s, which are `$k$th-nearest-neighbour'-type of quantities whose behaviour has been extensively studied in the literature \citep{Mack79,Hall83,Percus98,Evans02,Evans08}. Good properties for such quantities require the underlying density $f$ to be well-behaved in the following sense. 
\begin{assumption} \label{ass:f}  The density $f$ has convex compact support $C \subset \R^d$, with $\sup_{x,y \in C} \|x-y\| = c_1 < \infty$. It is bounded and bounded away from $0$ on $C$, i.e., there exist constants $a_1$ and $a_2$ such that $\inf_{x \in C} f(x) = a_1>0$ and $\sup_{x \in C} f(x) = a_2 <\infty$. In addition, $f$ is differentiable on $C$, with uniformly bounded partial derivatives of the first order.
\end{assumption} 

We have then the following result.

\begin{proposition} \label{prop:coefconsist} Under Assumptions \ref{ass:sample}-\ref{ass:f}, for all $j = j_0, \ldots, J$, $z \in \Z^d$ and $q \in Q$, the estimators (\ref{eqn:coefficients1}) and (\ref{eqn:coefficients2}) are such that
\begin{align*}
\E(\hat{\alpha}_{j,z}) =  \alpha_{j,z} + O(n^{-1/d}), \qquad & \var(\hat{\alpha}_{j,z}) = k^3\left(\frac{\Gamma(k)}{\Gamma(k+1/2)}\right)^2 O(n^{-1}) \\
\E(\hat{\beta}^{(q)}_{j,z}) =  \beta^{(q)}_{j,z} + O(n^{-1/d}), \qquad & \var(\hat{\beta}^{(q)}_{j,z}) = k^3\left(\frac{\Gamma(k)}{\Gamma(k+1/2)}\right)^2 O(n^{-1}),
\end{align*}
as $n \to \infty$. In particular, if $k$ is such that $k^{3/2}\frac{\Gamma(k)}{\Gamma(k+1/2)} = o(n^{1/2})$, then
\[\E\left(\left(\hat{\alpha}_{j,z}-\alpha_{j,z}\right)^2\right) \to 0 \qquad \text{ and } \qquad \E\left(\left(\hat{\beta}^{(q)}_{j,z}-\betaqjz\right)^2\right) \to 0 \]
as $n \to \infty$, and the estimators are consistent.
\end{proposition}
\begin{proof}
The proof makes use of an extension of Theorem 5.4 in \cite{Evans02}, and is given in Appendix.
\end{proof}
The condition $k^{3/2}\frac{\Gamma(k)}{\Gamma(k+1/2)} = o(n^{1/2})$ is obviously satisfied if $k$ keeps a fixed value. It also allows $k$ to grow along with $n$. As $k \to \infty$, $\Gamma(k)/\Gamma(k+1/2) \sim k^{-1/2}$ and the condition is equivalent to $k = o(n^{1/2})$. It appears that the (first order) asymptotic bias of $\hat{\alpha}_{j,z}$ and $\hat{\beta}^{(q)}_{j,z}$ does not depend on $k$, while their (first order) asymptotic variance increases with it. This can be attributed to larger covariances among the $V_{(k);i}$'s as $k$ gets large, and suggests -- at least at this level -- to keep $k$ as small as possible, that is, to use $k = 1$ always. By contrast, consistency of nonparametric density estimators built on $k$-Nearest-Neighbours ideas usually requires $k \to \infty$ as $n \to \infty$ \citep{Mack79,Hall83}. The fact that it seems here advantageous to keep $k$ as small as possible is, therefore, noteworthy. Below, the results are presented both for $k\doteq k_n$ satisfying $k^{3/2}\frac{\Gamma(k)}{\Gamma(k+1/2)} = o(n^{1/2})$ and for $k = 1$.

\section{Asymptotic properties of the estimators of $\sqrt{f}$ and $f$} \label{sec:gJ}

\subsection{Pointwise consistency}

In this subsection, the estimator $\hat{g}_J(x)$ (\ref{eqn:hatgJ})-(\ref{eqn:hatgJ2}) is first shown to be pointwise consistent for $\sqrt{f}(x)$ at all $x$. This essentially follows from the results of Section \ref{sec:coeff} through the theory of approximating kernels, see \cite{Bochner55} for early developments, and \cite{Meyer92} and \cite{Hardle98} for the wavelet case. From the father wavelet $\varphi$, let the {\it approximating kernel} $K: \R^d \times \R^d \to \R$ be
\begin{equation} K(x,y) = \sum_{z \in \Z^d} \varphi(x-z)\varphi(y-z) \label{eqn:kern} \end{equation}
and its {\it refinement} at resolution $j \in \N$ be
\[K_j(x,y) =  \sum_{z \in \Z^d} 2^{dj}\varphi(2^j x-z)\varphi(2^j y-z) = \sum_{z \in \Z^d} \varphi_{j,z}(x)\varphi_{j,z}(y). \]
Define the two associated operators:
\[K\phi (x) = \int_{\R^d} K(x,y) \phi(y)\diff y \qquad \text{ and } \qquad K_j\phi (x) = \int_{\R^d} K_j(x,y) \phi(y)\diff y, \]
for all functions $\phi \in L_2(\R^d)$. Then we have the following result.

\begin{proposition} \label{prop:gJbias} Under Assumptions \ref{ass:sample}-\ref{ass:f}, the estimator (\ref{eqn:hatgJ})-(\ref{eqn:hatgJ2}) is such that, at all $x \in C$,
\begin{align*} & (i) \qquad \E\left( \hat{g}_J(x)\right) = K_{J+1} \sqrt{f}(x)  + O(n^{-1/d}), \\
& (ii) \qquad \left(\frac{\Gamma(k+1/2)}{\Gamma(k)}\right)^2\,\frac{n}{k^{3}}\,\var\left(\hat{g}_J(x) \right) \leq \kappa  \int_{\R^d} K^2_{J+1}(x,y)\diff y +O(n^{-1/d}), \end{align*}
for some constant $\kappa < \infty$, as $n \to\infty$. Moreover, the order of the remainder terms holds uniformly in $x \in C$.
\end{proposition}
\begin{proof} See Appendix. \end{proof}

This result obviously implies the pointwise consistency of $\hat{g}_J(x)$ for $\sqrt{f}(x)$ at any fixed $x \in C$ provided that $k^{3/2}\frac{\Gamma(k)}{\Gamma(k+1/2)} = o(n^{1/2})$, in particular if $k$ is kept fixed.

\subsection{Uniform $L_2$-consistency}

Consistency in Mean Integrated Squared Error ($L_2$-consistency) of estimator (\ref{eqn:hatgJ})-(\ref{eqn:hatgJ2}) can now be established uniformly over large classes of functions, such as Sobolev classes. Call $W^{m,p}(\Omega)$ the Sobolev space of functions defined on $\Omega \subset \R^d$ for which all mixed partial derivatives up to order $m \geq 0$ exist (in the weak sense) and belong to $L_p(\Omega)$, $1 \leq p \leq \infty$. Formally,
\begin{equation*} W^{m,p}(\Omega) = \left \{ \phi \in L^p(\Omega) : D^{\alpha} \phi \in L^p(\Omega) \,\, \forall \alpha \in \N^d: |\alpha| \leqslant m \right \},  \end{equation*}
where $D^\alpha$ is the $\alpha^\text{th}$ (multi-index notation) partial weak derivative operator, and $|\alpha| = \sum_{k=1}^d \alpha_k$. A norm on $W^{m,p}(\Omega)$ is classically defined as $\|\phi\|_{m,p} = \sum_{|\alpha|\leq m} \|D^\alpha \phi \|_p$ \citep{Triebel1992}.

\ppn It follows from Assumption \ref{ass:f} that there exists an integer $m \geq 1$ such that $f \in W^{m,2}(C)$: $f$ has uniformly bounded partial derivatives on $C$, which implies $f \in W^{1,\infty}(C)$, and as $W^{1,\infty}(C) \subset W^{1,2}(C)$, at least $m=1$. Of course, more regular (i.e.\ smoother) densities $f$ allow for a higher value of $m$. In addition, under Assumption \ref{ass:f}, $\sqrt{f} \in W^{m,2}(C)$ as well. This appears clearly from the multivariate version of Fa\`a di Bruno's formula (see e.g.\ \cite{Hardy06}), which reads here, for all $\alpha \in \N^d$ such that $|\alpha|\leq m$:
\[D^\alpha\sqrt{f} = \sum_{\xi\in\Xi} f^{1/2 - \left|\xi\right|} \prod_{\beta\in\xi}D^\beta f,\]
where $\Xi$ is the set of all partitions $\xi$ of the elements of $\alpha$ and the product is over all `blocks' $\beta$ of the partition $\xi$. Then the $L^2$-norm of the second factor in each term is bounded because $|\beta| \leq m$ and $f \in W^{m,2}(C)$, and the first factor $f^{1/2 - \left|\xi\right|}$ is uniformly bounded for all $0 \le \left|\xi\right| \le m$, because $f$ is both bounded from above (case $\left|\xi\right|=0$) and bounded away from $0$ (case $\left|\xi\right| \ge 1$). This also implies that, if $f \in B^{m,2}(L)= \{ \phi \in W^{m,2}(C): \|\phi\|_{m,2} \leq L\}$ for some constant $ 0 \leq L <\infty$, i.e., a ball of radius $L$ in $W^{m,2}(C)$, then $\sqrt{f} \in B^{m,2}(L')$ for some other constant $0 \leq L' < \infty$.

\ppn Now, suppose that the father wavelet $\varphi$ introduced in Assumption \ref{ass:wav} is such that the induced kernel (\ref{eqn:kern}) satisfies the following assumption.

\begin{assumption} \label{ass:kern} The kernel $K$ (\ref{eqn:kern}) is such that $|K(x,y)| \leq F(x-y)$, for some square integrable function $F: \R^d \to \R$ with $\int_{\R^d} |x|^{\nu}F(x)\diff x < \infty$ for all $\nu \in \N^d$ such that $|\nu| = m$. Moreover, for all $x\in \R^d$, $\int_{\R^d} (y-x)^{\nu'} K(x,y)\diff y = \delta_{0,\nu'}$, for all $\nu' \in \N^d$ such that $|\nu'| \leq m-1$.
\end{assumption}

Here, for $x \in \R^d$ and $\nu \in \N^d$, $|x|^\nu = \prod_{k=1}^d |x_k|^{\nu_k}$, and $\delta_{\nu,\nu'}$ is the $d$-fold Kronecker delta, equal to 1 if $\nu_k = \nu'_k$ $\forall k\in \{1,\ldots,d\}$ and 0 otherwise. Then, one can prove the following.

\begin{theorem} \label{thm:uniconst} Under Assumptions \ref{ass:sample}-\ref{ass:f} and Assumption \ref{ass:kern}, the estimator (\ref{eqn:hatgJ})-(\ref{eqn:hatgJ2}) is such that
\begin{equation} \sup_{f \in B^{m,2}(L)} \E\left(\|\hat{g}_J - \sqrt{f} \|_2^2\right)  \leq   \kappa_1 2^{-2Jm} + \kappa_2 n^{-2/d} + \kappa'_3 n^{-1} k^3 \left(\frac{\Gamma(k)}{\Gamma(k+1/2)} \right)^2 2^{dJ}, \label{eqn:uniL2}\end{equation}
for some constants $\kappa_1,\kappa_2,\kappa'_3 < \infty$ and $n$ large enough.
\end{theorem}
\begin{proof} See Appendix. \end{proof}

\ppn Clearly, the bound in the right-hand side of (\ref{eqn:uniL2}) is a non-decreasing function of $k$, which suggests to take $k=1$ as it was already noted below Proposition \ref{prop:coefconsist}. For that choice, we have directly:
\begin{corollary} Under Assumptions \ref{ass:sample}-\ref{ass:f} and Assumption \ref{ass:kern}, the estimator (\ref{eqn:hatgJ})-(\ref{eqn:hatgJ2}) with $k=1$ in (\ref{eqn:coefficients1})-(\ref{eqn:coefficients2}) is such that
 \begin{equation*} \sup_{f \in B^{m,2}(L)} \E\left(\|\hat{g}_J - \sqrt{f} \|_2^2\right) \leq   \kappa_1 2^{-2Jm} + \kappa_2 n^{-2/d} + \kappa_3 \frac{2^{dJ}}{n}, \end{equation*}
for some constants $\kappa_1,\kappa_2,\kappa_3 < \infty$ and $n$ large enough.
\end{corollary}

\ppn The terms depending on $J$ are balanced for $2^J \propto n^{\frac{1}{2m + d }}$, in which case 
\[\sup_{f \in B^{m,2}(L)} \E\left(\|\hat{g}_J - \sqrt{f} \|_2^2\right)  \leq   \kappa' n^{-\frac{2m}{2m+d}} + \kappa'' n^{-2/d},\]
for two constants $\kappa', \kappa'' < \infty$. Finally, by the Cauchy-Schwartz inequality, 
\[\|\hat{g}_J^2-f\|_2^2 = \|(\hat{g}_J-\sqrt{f})(\hat{g}_J+\sqrt{f})\|_2^2 \leq \|\hat{g}_J-\sqrt{f}\|_2^2 \times \|\hat{g}_J+\sqrt{f}\|_2^2.\]
Assumptions \ref{ass:wav} and \ref{ass:f} ensure that the second factor is bounded, whereby we have the following result about $\hat{g}^2_J$ as an estimator of the density $f$.
\begin{theorem} Under Assumptions \ref{ass:sample}-\ref{ass:f} and Assumption \ref{ass:kern}, the estimator $\hat{g}_J^2$ with $k=1$ in (\ref{eqn:coefficients1})-(\ref{eqn:coefficients2})  and $2^J \propto n^{\frac{1}{2m + d }}$ is a uniformly $L_2$-consistent estimator of $f$, such that
 \begin{equation} \sup_{f \in B^{m,2}(L)} \E\left(\|\hat{g}^2_J - f \|_2^2\right)  \leq   \kappa' n^{-\frac{2m}{2m+d}} + \kappa'' n^{-2/d}, \label{eqn:L2rate} \end{equation}
for some constants $\kappa', \kappa'' < \infty$.
\end{theorem}

\ppn Note that the first term in the right-hand side of (\ref{eqn:L2rate}) is the optimal nonparametric rate of convergence in this situation, as per \cite{Stone82}'s classical results. That term is dominated by the second one only for $d > \frac{2m}{m-1}$. Hence we have the following corollary.
\begin{corollary}
Under Assumptions \ref{ass:sample}-\ref{ass:f} and Assumption \ref{ass:kern}, the estimator $\hat{g}_J^2$ with $k=1$ in (\ref{eqn:coefficients1})-(\ref{eqn:coefficients2}) and $2^J \propto n^{\frac{1}{2m + d }}$ is asymptotically optimal for $f$ uniformly over $B^{m,2}(L) \subset W^{m,2}(C)$, in the sense that 
\[\sup_{f \in B^{m,2}(L)} \E\left(\|\hat{g}^2_J - f \|_2^2\right) \leq \kappa' n^{-\frac{2m}{2m+d}}, \]
for $d \leq \frac{2m}{m-1}$.
\end{corollary}
As $\frac{2m}{m-1} > 2$, the estimator is always optimal in one and two dimensions. Under the classical mild smoothness assumption $m=2$, it is optimal for $1 \leq d \leq 4$ -- this probably covers most of the cases of practical interest, given that the optimal rate of convergence itself becomes very poor in higher dimensions ({\it Curse of Dimensionality}, \cite{Geenens11}). In any case, for `rough' densities $f$ ($m=1$), the estimator reaches the optimal rate in all dimensions.

\section{Numerical experiments} \label{sec:numE}

\subsection{Simulation study} 

In this section the practical performance of the shape-preserving estimator $\hat{g}_J^2$ based on (\ref{eqn:hatgJ})-(\ref{eqn:hatgJ2}) is compared to that of the classical wavelet estimator. Three bivariate ($d=2$) Gaussian mixtures were considered: (a) two components, showing two peaks with very different covariance structures (Figure \ref{fig:mult}); (b) two components, showing two similar peaks (Figure \ref{fig:mix2}), and (c) a bivariate version of \cite{Marron1992}'s `smooth comb', showing 4 peaks of decreasing spread (Figure \ref{fig:mix3}).\footnote{The exact expressions are available from the authors upon request.} Those where scaled and truncated to the unit square $[0,1]^2$, in order to satisfy Assumption \ref{ass:f}. Note that mixtures (a) and (c) exhibit peaks of different spread and orientation, features known to cause difficulties in density estimation.

\begin{figure}[H] \label{fig:true_figures}
	\begin{subfigure}[t]{0.33\textwidth}
		\includegraphics[width=\textwidth]{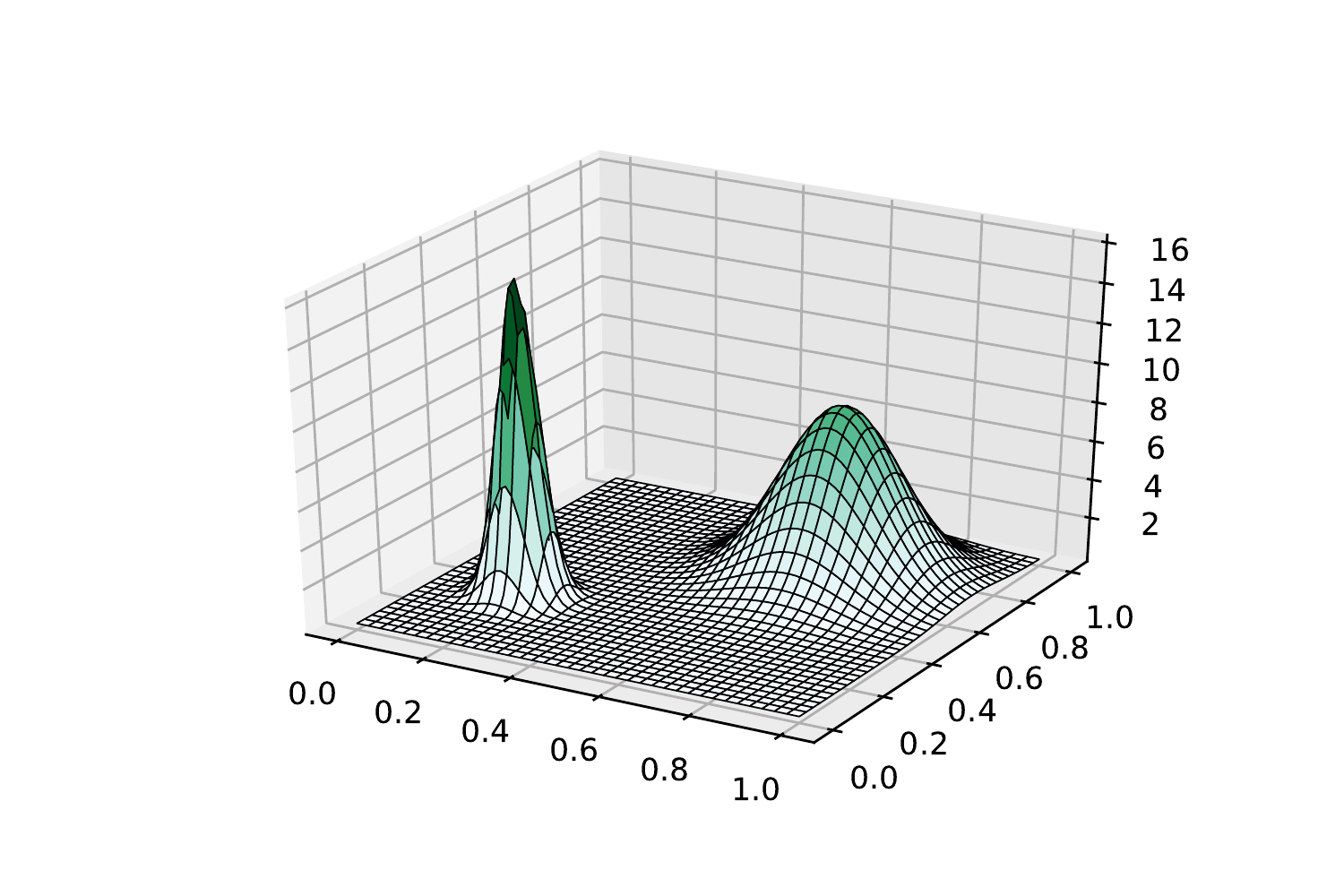}
		\caption{}
		\label{fig:mult}
	\end{subfigure}
	\begin{subfigure}[t]{0.33\textwidth}
		\includegraphics[width=\textwidth]{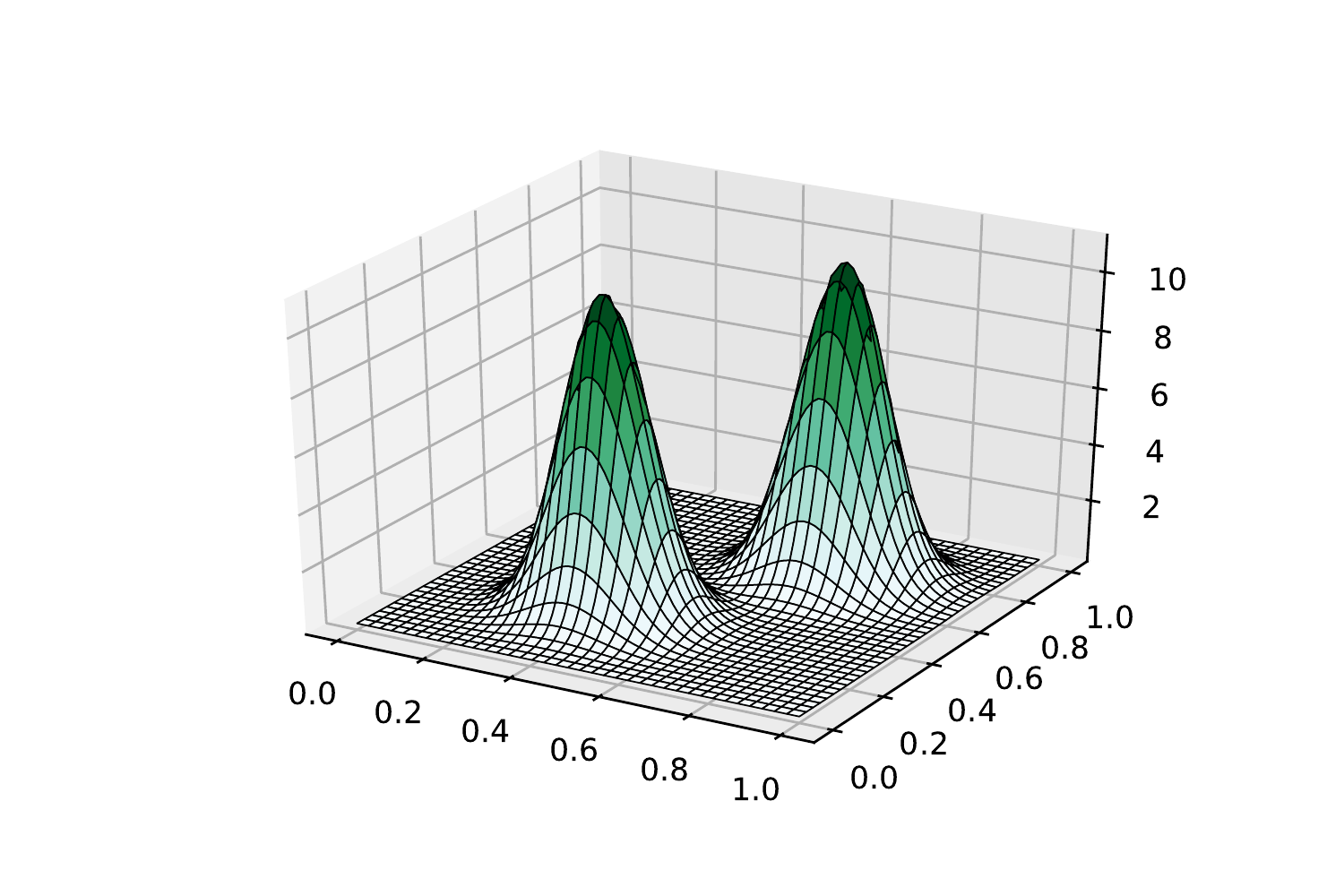}
		\caption{}
		\label{fig:mix2}
	\end{subfigure}
	\begin{subfigure}[t]{0.33\textwidth}
		\includegraphics[width=\textwidth]{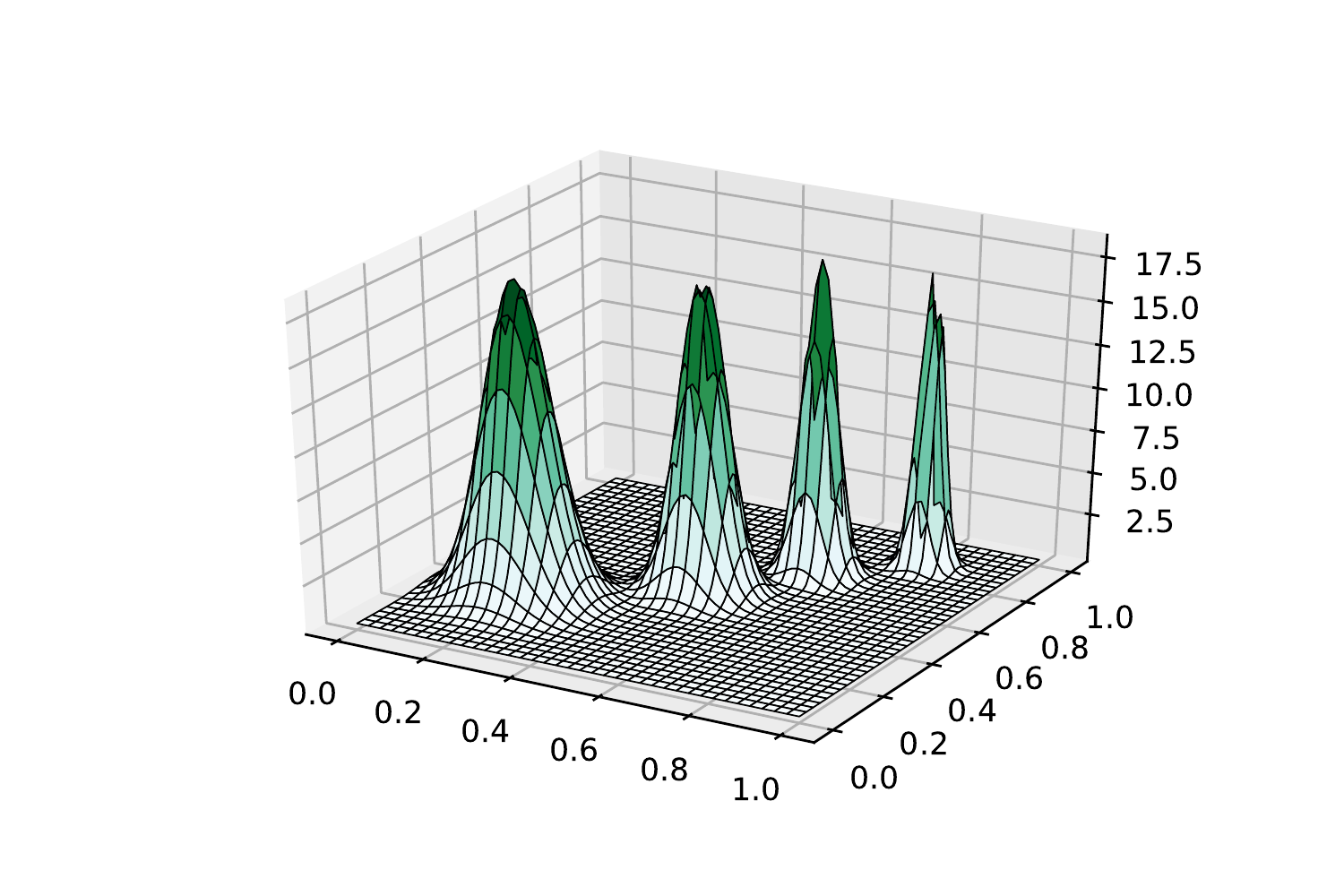}
		\caption{}
		\label{fig:mix3}
	\end{subfigure}
	\caption{Bivariate densities used in the simulation study.}
\end{figure}

\ppn For each density, $M=500$ random samples of size $n=2^\ell$, for $\ell \in \{7,\ldots,12\}$, that is, from $n=128$ up to $n=4096$,\footnote{Sample sizes as powers of 2 are customary in the wavelet framework due to their suitability when resorting to the Fast Wavelet Transform, however the estimator described in Section \ref{subsec:def} remains obviously valid for any arbitrary sample size $n$.} were generated, and our procedure was used on each of them for estimating $f$. Proper normalisation of all estimates was enforced through (\ref{eqn:unitint}). The accuracy of a given estimate $\hat{f}$ was measured by the Integrated Squared Error (ISE) $\int_{[0,1]^2} \left(\hat{f}(x)-f(x) \right)^2\diff x$, approximated by Riemannian summing on a fine regular partition of $[0,1]^2$. The Mean Integrated Squared Error (MISE) of an estimator was then approximated by averaging the ISE's over the $M=500$ Monte-Carlo replications, see Table \ref{tab:mise}. 

\ppn Estimators (\ref{eqn:coefficients1})-(\ref{eqn:coefficients2}) were computed with bivariate wavelets $\varphi_{j,z}$ and $\psi_{j,z}^{(q)}$ obtained by tensor products of univariate Daubechies wavelets with 6 vanishing moments \citep{Daubechies92}. In agreement with the asymptotic results, the value $k = 1$ in (\ref{eqn:coefficients1})-(\ref{eqn:coefficients2}) was given primary focus, but $k = 2, 4, 8, .., \sqrt{n}$ were also tested to investigate the effect of $k$ in finite samples. For the three densities and all sample sizes, the choice $k=1$ always lead to the final estimator with the smallest MISE, or within statistical significance (given $M=500$ Monte-Carlo replications) to the estimator with the smallest MISE. Hence in Table \ref{tab:mise} only the results for $k=1$ are reported. In (\ref{eqn:hatgJ}), the baseline resolution was taken $j_0 = 0$ and the resolution levels $J \in \{-1, 0, 1, 2,3\}$ were considered -- the case $J =  - 1$ is here defined as the estimator with the trend at baseline level $j_0=0$ only. For comparison, the density $f$ was also estimated on each sample by the classical wavelet estimator described in \cite{Hardle98}, whose MISE was approximated in the exact same way as above.

\ppn The whole procedure was developed in Python, using the BallTree $k$-Nearest neighbour algorithm \citep{Omohundro89} and the PyWavelets library that supports a number of orthogonal and biorthogonal wavelet families. It is available as open source in a github repository\footnote{{\tt https://github.com/carlosayam/PyWDE}} along with an implementation of the classic wavelet estimator. Note that, despite only the case $d=2$ is reported here, the estimator can handle potentially any number of dimensions.

\begin{table}[ht!]
	\centering
	\begin{minipage}[t]{0.31\textwidth}
		\centering
		\textbf{Gaussian mix (a)} \\
		\vspace{6pt}
		\begin{tabular}{|r|r|r|r|}
			\hline
			\textbf{$n$} & \textbf{$J+1$} & \textbf{SP} & \textbf{Class.} \\
			\hline
			\multirow{6}{*}{128}
			& 0 & 3.490
			& 3.686 \\
			\cline{2-4}
			& 1 & 2.907
			& 3.097 \\
			\cline{2-4}
			& 2 & 1.358
			& 1.086 \\
			\cline{2-4}
			& \textbf{3} & \textbf{1.199}
			& \textbf{0.862} \\
			\cline{2-4}
			& 4 & 4.697
			& 1.964 \\
			\hline
			\multirow{6}{*}{256}
			& 0 & 3.491
			& 3.686 \\
			\cline{2-4}
			& 1 & 2.891
			& 3.092 \\
			\cline{2-4}
			& 2 & 1.286
			& 1.043 \\
			\cline{2-4}
			& \textbf{3} & \textbf{0.778}
			& \textbf{0.634} \\
			\cline{2-4}
			& 4 & 2.351
			& 0.995 \\
			\hline
			\multirow{6}{*}{512}
			& 0 & 3.491
			& 3.686 \\
			\cline{2-4}
			& 1 & 2.880
			& 3.090 \\
			\cline{2-4}
			& 2 & 1.235
			& 1.022 \\
			\cline{2-4}
			& \textbf{3} & \textbf{0.543}
			& 0.523 \\
			\cline{2-4}
			& 4 & 1.093
			& \textbf{0.518} \\
			\hline
			\multirow{6}{*}{1024}
			& 0 & 3.491
			& 3.686 \\
			\cline{2-4}
			& 1 & 2.873
			& 3.088 \\
			\cline{2-4}
			& 2 & 1.211
			& 1.012 \\
			\cline{2-4}
			& \textbf{3} & \textbf{0.406}
			& 0.468 \\
			\cline{2-4}
			& 4 & 0.529
			& \textbf{0.274} \\
			\hline
			\multirow{6}{*}{2048}
			& 0 & 3.491
			& 3.686 \\
			\cline{2-4}
			& 1 & 2.872
			& 3.087 \\
			\cline{2-4}
			& 2 & 1.190
			& 1.007 \\
			\cline{2-4}
			& 3 & 0.343
			& 0.439 \\
			\cline{2-4}
			& \textbf{4} & \textbf{0.267}
			& \textbf{0.149} \\
			\hline
			\multirow{6}{*}{4096}
			& 0 & 3.492
			& 3.686 \\
			\cline{2-4}
			& 1 & 2.871
			& 3.087 \\
			\cline{2-4}
			& 2 & 1.180
			& 1.004 \\
			\cline{2-4}
			& 3 & 0.312
			& 0.425 \\
			\cline{2-4}
			& \textbf{4} & \textbf{0.134}
			& \textbf{0.090} \\
			\hline
		\end{tabular}
	\end{minipage}
	\begin{minipage}[t]{0.31\textwidth}
		\centering
		\textbf{Gaussian mix (b)} \\
		\vspace{6pt}
		\begin{tabular}{|r|r|r|r|}
			\hline
			\textbf{$n$} & \textbf{$J+1$} & \textbf{SP} & \textbf{Class.} \\
			\hline
			\multirow{6}{*}{128}
			& 0 & 4.315
			& 4.383 \\
			\cline{2-4}
			& 1 & 3.269
			& 3.436 \\
			\cline{2-4}
			& \textbf{2} & \textbf{0.836}
			& 1.225 \\
			\cline{2-4}
			& 3 & 0.912
			& \textbf{0.540} \\
			\cline{2-4}
			& 4 & 4.337
			& 1.909 \\
			\hline
			\multirow{6}{*}{256}
			& 0 & 4.314
			& 4.384 \\
			\cline{2-4}
			& 1 & 3.253
			& 3.430 \\
			\cline{2-4}
			& 2 & 0.747
			& 1.184 \\
			\cline{2-4}
			& \textbf{3} & \textbf{0.471}
			& \textbf{0.326} \\
			\cline{2-4}
			& 4 & 2.143
			& 0.989 \\
			\hline
			\multirow{6}{*}{512}
			& 0 & 4.314
			& 4.383 \\
			\cline{2-4}
			& 1 & 3.247
			& 3.427 \\
			\cline{2-4}
			& 2 & 0.714
			& 1.165 \\
			\cline{2-4}
			& \textbf{3} & \textbf{0.245}
			& \textbf{0.210} \\
			\cline{2-4}
			& 4 & 1.048
			& 0.496 \\
			\hline
			\multirow{6}{*}{1024}
			& 0 & 4.314
			& 4.383 \\
			\cline{2-4}
			& 1 & 3.243
			& 3.425 \\
			\cline{2-4}
			& 2 & 0.690
			& 1.152 \\
			\cline{2-4}
			& \textbf{3} & \textbf{0.137}
			& \textbf{0.151} \\
			\cline{2-4}
			& 4 & 0.519
			& 0.245 \\
			\hline
			\multirow{6}{*}{2048}
			& 0 & 4.313
			& 4.383 \\
			\cline{2-4}
			& 1 & 3.242
			& 3.424 \\
			\cline{2-4}
			& 2 & 0.680
			& 1.147 \\
			\cline{2-4}
			& \textbf{3} & \textbf{0.081}
			& \textbf{0.123} \\
			\cline{2-4}
			& 4 & 0.261
			& 0.123 \\
			\hline
			\multirow{6}{*}{4096}
			& 0 & 4.313
			& 4.383 \\
			\cline{2-4}
			& 1 & 3.241
			& 3.424 \\
			\cline{2-4}
			& 2 & 0.676
			& 1.145 \\
			\cline{2-4}
			& \textbf{3} & \textbf{0.051}
			& 0.110 \\
			\cline{2-4}
			& 4 & 0.128
			& \textbf{0.061} \\
			\hline
		\end{tabular}
	\end{minipage}
	\begin{minipage}[t]{0.31\textwidth}
		\centering
		\textbf{Comb (c)} \\
		\vspace{6pt}
		\begin{tabular}{|r|r|r|r|}
			\hline
			\textbf{$n$} & \textbf{$J+1$} & \textbf{SP} & \textbf{Class.} \\
			\hline
			\multirow{6}{*}{128}
			& 0 & 8.320
			& 8.324 \\
			\cline{2-4}
			& 1 & 6.800
			& 6.818 \\
			\cline{2-4}
			& 2 & 4.565
			& 5.142 \\
			\cline{2-4}
			& \textbf{3} & \textbf{1.972}
			& \textbf{2.082} \\
			\cline{2-4}
			& 4 & 3.743
			& 2.159 \\
			\hline
			\multirow{6}{*}{256}
			& 0 & 8.319
			& 8.323 \\
			\cline{2-4}
			& 1 & 6.793
			& 6.813 \\
			\cline{2-4}
			& 2 & 4.480
			& 5.096 \\
			\cline{2-4}
			& \textbf{3} & \textbf{1.561}
			& 1.861 \\
			\cline{2-4}
			& 4 & 1.973
			& \textbf{1.213} \\
			\hline
			\multirow{6}{*}{512}
			& 0 & 8.319
			& 8.323 \\
			\cline{2-4}
			& 1 & 6.790
			& 6.811 \\
			\cline{2-4}
			& 2 & 4.425
			& 5.074 \\
			\cline{2-4}
			& 3 & 1.332
			& 1.756 \\
			\cline{2-4}
			& \textbf{4} & \textbf{1.087}
			& \textbf{0.745} \\
			\hline
			\multirow{6}{*}{1024}
			& 0 & 8.320
			& 8.323 \\
			\cline{2-4}
			& 1 & 6.790
			& 6.809 \\
			\cline{2-4}
			& 2 & 4.395
			& 5.064 \\
			\cline{2-4}
			& 3 & 1.203
			& 1.700 \\
			\cline{2-4}
			& \textbf{4} & \textbf{0.585}
			& \textbf{0.499} \\
			\hline
			\multirow{6}{*}{2048}
			& 0 & 8.320
			& 8.323 \\
			\cline{2-4}
			& 1 & 6.789
			& 6.809 \\
			\cline{2-4}
			& 2 & 4.380
			& 5.058 \\
			\cline{2-4}
			& 3 & 1.141
			& 1.673 \\
			\cline{2-4}
			& \textbf{4} & \textbf{0.349}
			& \textbf{0.378} \\
			\hline
			\multirow{6}{*}{4096}
			& 0 & 8.320
			& 8.323 \\
			\cline{2-4}
			& 1 & 6.787
			& 6.809 \\
			\cline{2-4}
			& 2 & 4.370
			& 5.055 \\
			\cline{2-4}
			& 3 & 1.103
			& 1.659 \\
			\cline{2-4}
			& \textbf{4} & \textbf{0.225}
			& \textbf{0.317} \\
			\hline
		\end{tabular}
	\end{minipage}
	\caption{(Approximated) MISE of the shape-preserving estimator (SP) and the classical wavelet estimator (Class.) for different sample sizes and different values of $J+1$ ($j_0=0$). The smallest MISE is highlighted for each sample size.}
	\label{tab:mise}
\end{table}

\ppn Analysing Table \ref{tab:mise} reveals that neither estimator seems to have an absolute edge over the other, and the observed differences in MISE are low. For small sample sizes, the classical estimator is usually doing slightly better (although not always). This can be understood as it is based on simple averages which typically behave better than nearest-neighbour distances when the number of observations is not large. On the other hand, for larger samples, the Shape-Preserving (SP) estimator does usually better (although not always). It profits from the fact that it makes proper use of the probability mass that the classical one loses below zero in the low-density areas. This is illustrated by Figure \ref{fig:comparison_figures}, which shows typical estimates for the shape-preserving estimator and the classical one for sample size $n=4096$ ($k=1$, $j_0=0$ and $J=3$). Note how the classic estimator loses mass in areas of low density, even for this large sample. Therefore, although the results in Table \ref{tab:mise} indicate that the classical estimator might be slightly more accurate {\it sensu stricto} ($\text{MISE}=0.090$ for classical, $\text{MISE}=0.134$ for SP), it seems that the SP estimator may be preferable: the `price to pay' (in terms of MISE) for getting estimates which are automatically proper densities is quite low. 

\begin{figure}[!t]
	\centering
	\begin{subfigure}{\linewidth}
		\centering
		\includegraphics[width=0.55\linewidth]{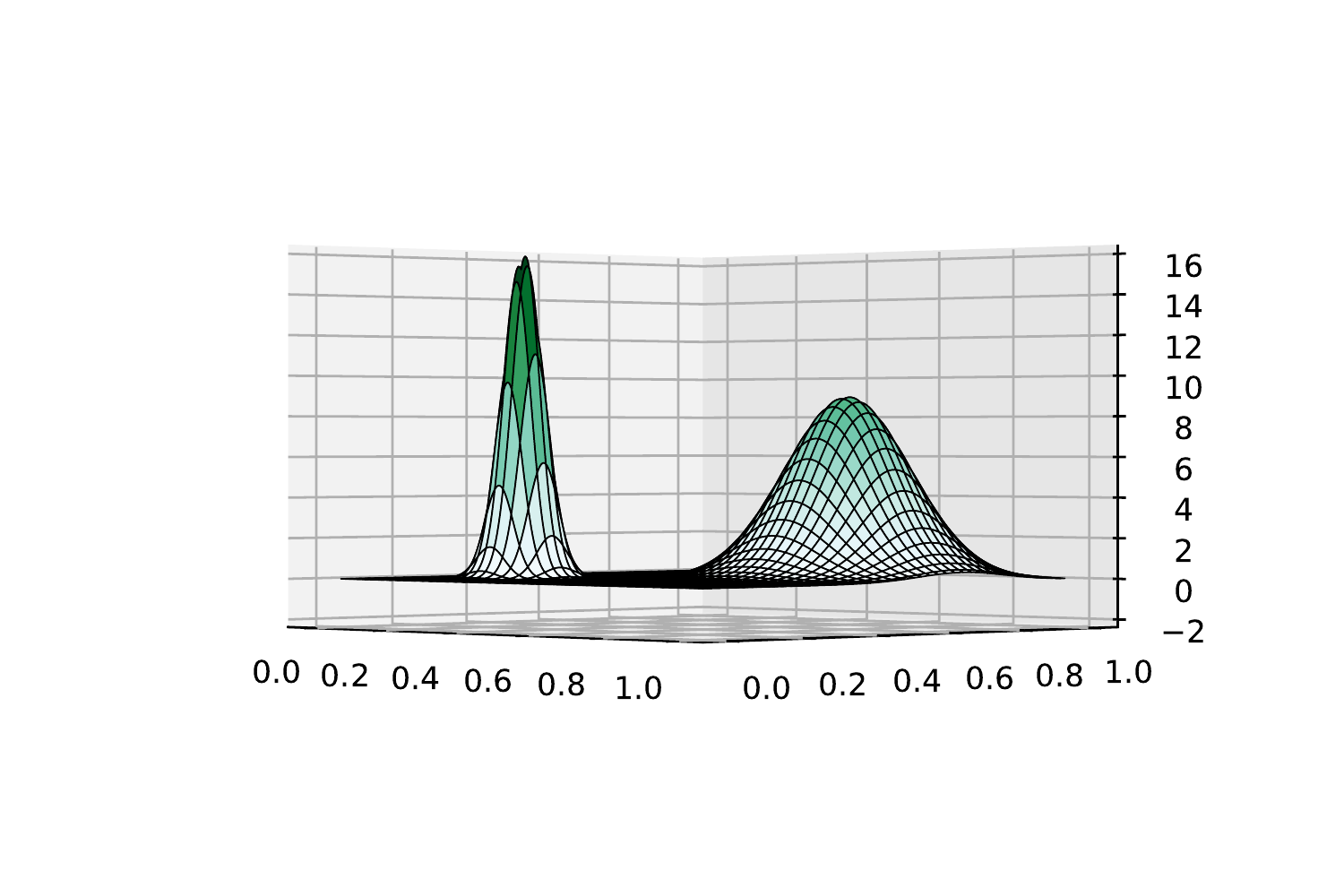}
		\caption{True density}
		\label{fig:comp_true}
	\end{subfigure}
	\begin{subfigure}{\linewidth}
		\centering
		\includegraphics[width=0.55\linewidth]{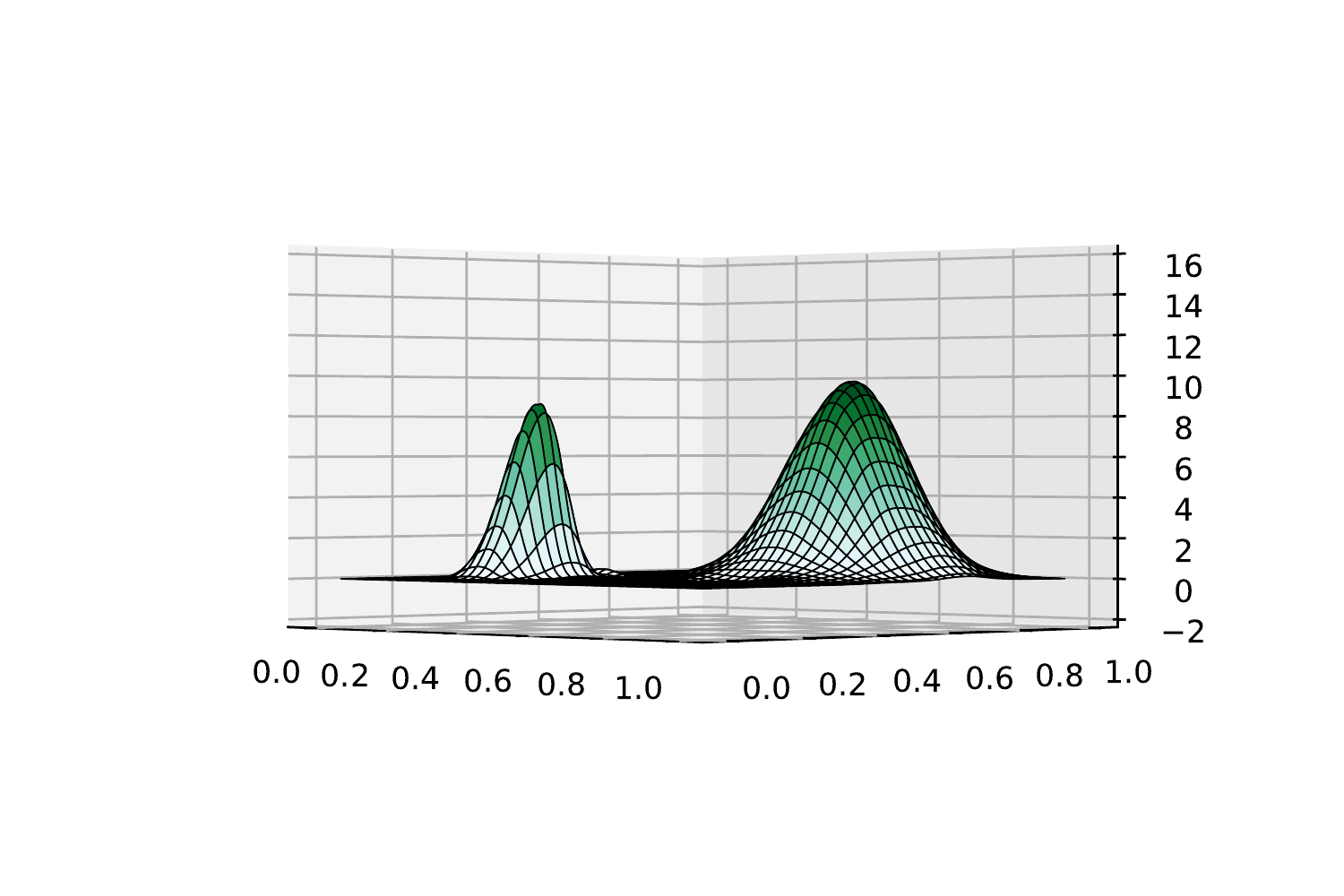}
		\caption{Shape preserving estimator}
		\label{fig:new_wde}
	\end{subfigure}
	\begin{subfigure}{\linewidth}
		\centering
		\includegraphics[width=0.55\linewidth]{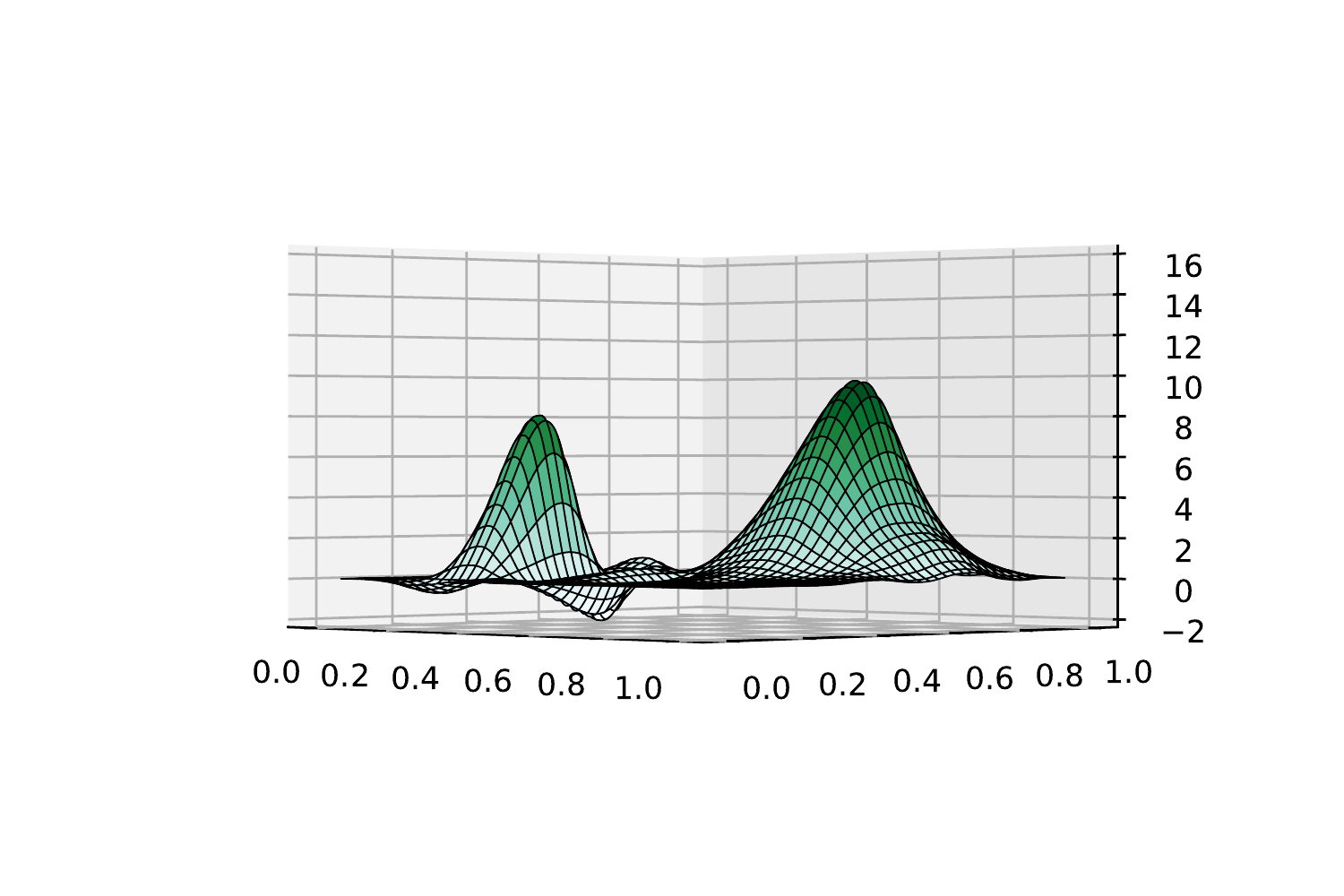}
		\caption{Classical estimator}
		\label{fig:classic_wde}
	\end{subfigure}
	\caption{Comparisons of estimates for Gaussian Mixture (a), $n=4096$, $k=1$ and $j_0=0$ and $J=3$.}
	\label{fig:comparison_figures}
\end{figure}


\ppn Figure \ref{fig:comparison_figures} also reveals how challenging it is, for both estimators, to re-construct two peaks of such different spread. In that respect, the introduction of a thresholding scheme would be helpful to allow a higher resolution to be selected while killing out any unwarranted noise. The shape-preserving estimator is expected to profit more from the introduction of such thresholding, as it is noted from Table \ref{tab:mise} that the classical estimator sometimes allows a higher resolution, already. More on this in Section \ref{sec:ccl}.

\subsection{Real data: Old Faithful geyser} \label{subsec:oldfaithful}

Old Faithful geyser is a very active geyser in the Yellowstone National Park, Wyoming, USA.\footnote{see {\tt www.geyserstudy.org/geyser.aspx?pGeyserNo=OLDFAITHFUL}.} Data on eruption times and waiting times (both in minutes) between eruptions of Old Faithful form a well-known bivariate data set of $n=272$ observations. In particular, it was used for illustration in \cite{Vannucci95}, in a review of different types of wavelet density estimators. The shape-preserving estimator was computed on these data using Daubechies wavelets with 7 vanishing moments (as in \cite{Vannucci95}). The best results were obtained with $j_0=0$ and $J=2$, producing the estimate shown in Figure \ref{fig:oldfaithful}. As opposed to Figure 6 in  \cite{Vannucci95}, the shape-preserving estimator shows some small bumps of potential interest near the main peaks. In view of the raw data (scatter plot, left panel) and other available kernel-based density estimates \citep{Silverman1986,Hyndman96}, this seems legitimate.

\begin{figure}[t!] 
	\begin{subfigure}[t]{0.5\textwidth}
		\includegraphics[width=\textwidth]{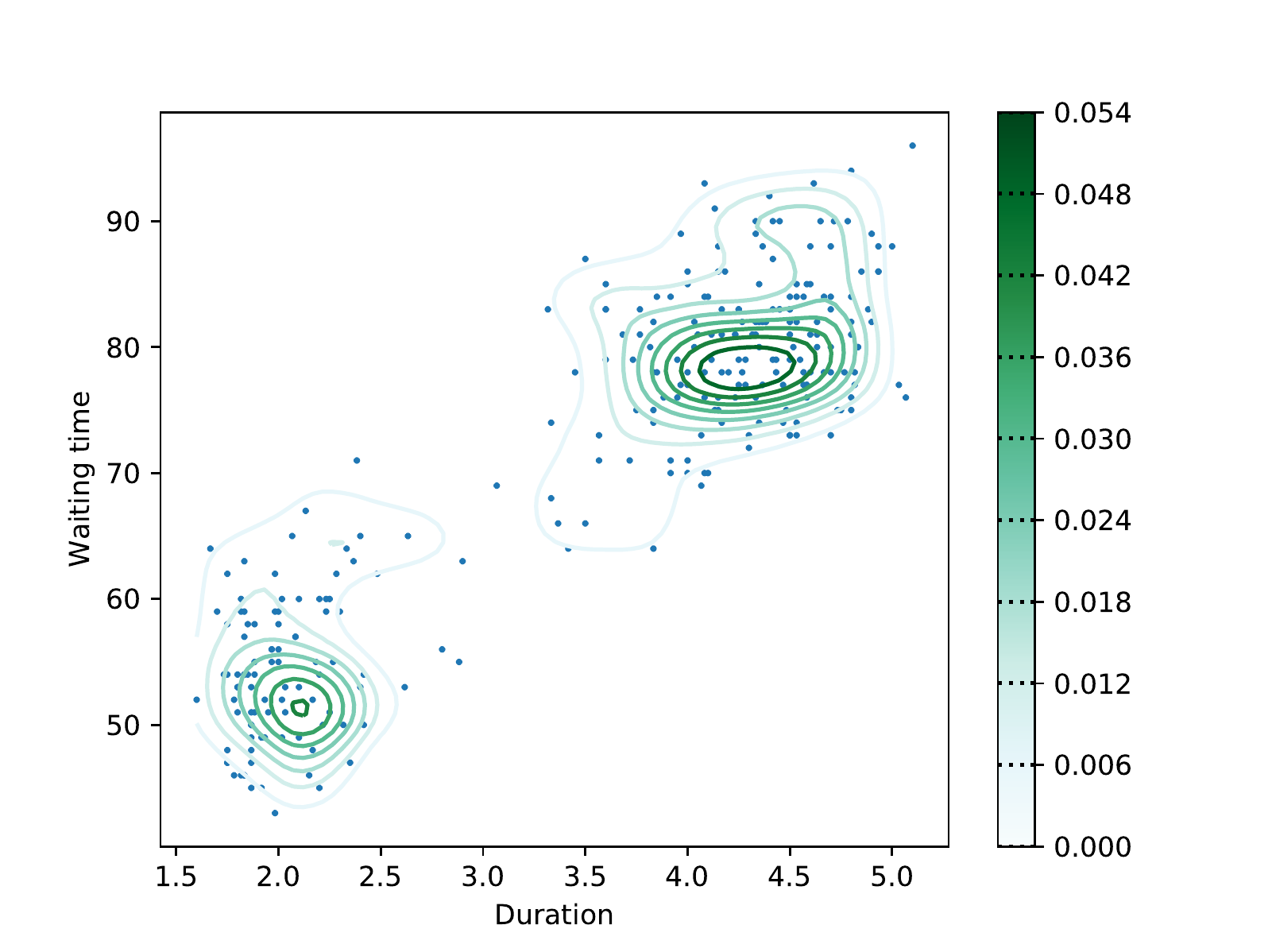}
		\caption{Scatter and contour plot}
		\label{fig:oldf-data}
	\end{subfigure}
	\hfill
	\begin{subfigure}[t]{0.5\textwidth}
		\includegraphics[width=\textwidth]{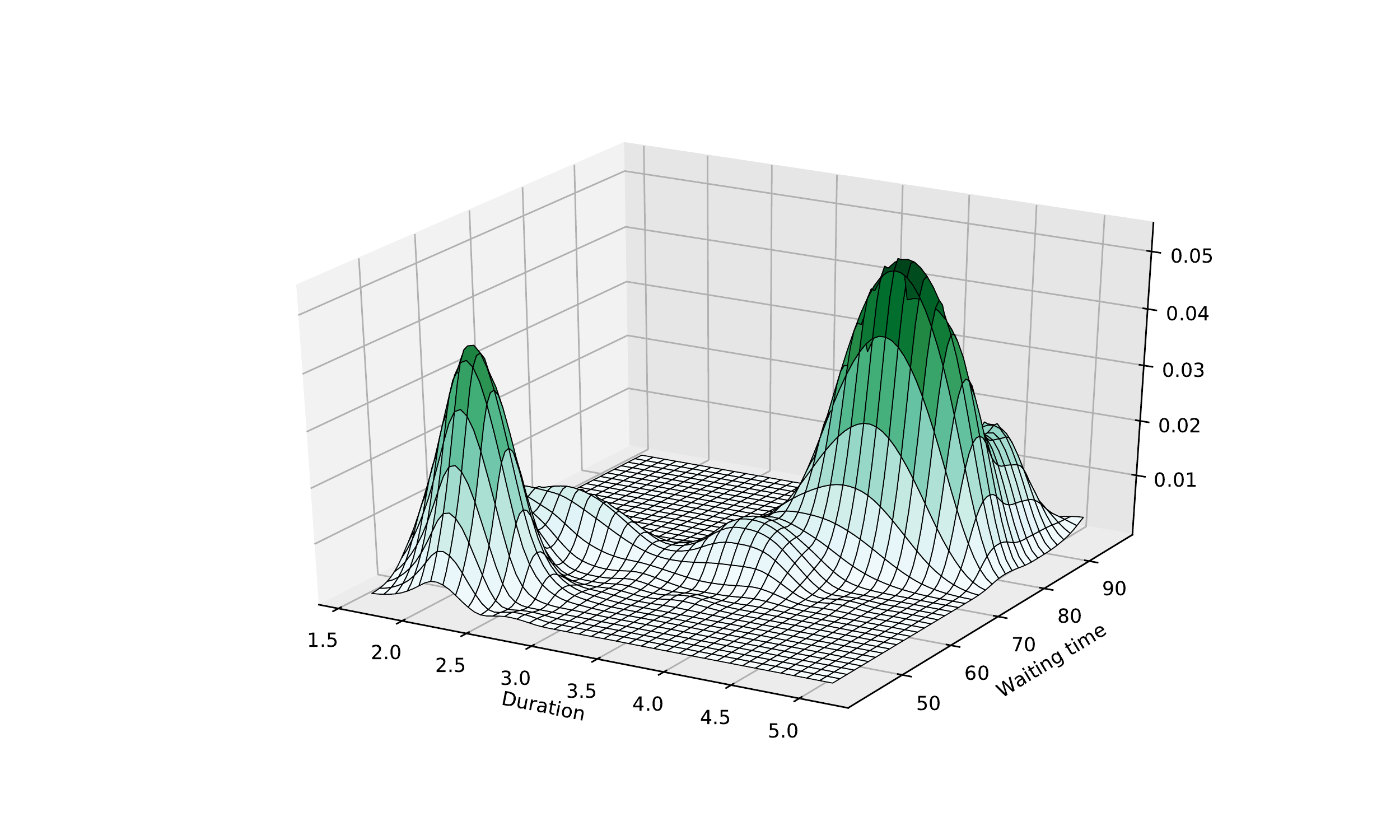}
		\caption{3D-density estimate}
		\label{fig:oldf-wde}
	\end{subfigure}
	\caption{Old Faithful dataset}
	\label{fig:oldfaithful}
\end{figure}

\section{Conclusions and future work}  \label{sec:ccl}

\cite{Penev97} suggested an elegant construction of a wavelet estimator of the square-root of a univariate probability density in order to deal with negativity issues in an automatic way. Based on spacings, their idea could not be easily generalised beyond the univariate case, though. This paper provides such an extension, essentially making use of nearest-neighbour-balls, the ``{\it probabilistic counterpart to univariate spacings}'' \citep{Ranneby05} in higher dimensions. The asymptotic properties of the estimator were obtained. It always attains the optimal rate of convergence in Mean Integrated Square Error in $d=1$ and $d=2$ dimensions, in dimensions up to $d=4$ for reasonably smooth densities, and in all dimensions for `rough' densities. In practice, the estimator was seen to be on par with the classical wavelet estimator, while automatically producing estimates which are always {\it bona fide} densities.

\ppn Continuation of this research includes the introduction of a thresholding scheme. It is well-known that thresholding wavelet coefficients in the classical case gives better estimates in general Besov spaces \citep{DJKP93,DonohoJohnstone98}. For a set of coefficients $\{c_z;z \in \Zd\}$ essentially defining a particular wavelet family, the father wavelet satisfies $\varphi(x) = \sum_{z \in \Z^d} c_{z} \varphi\left(2 x - z\right)$ (and similar for the functions $\psi^{(q)}$'s); see \citet{Daubechies92}. This implies that $\varphi_{j,z}(x) = \sum_{z' \in \Z^d} c_{z'} \varphi_{j+1,z' - 2 z}\left(x\right)$, which, in turn, carries over to the wavelet coefficients, viz.\ $\alpha_{j,z} = \sum_{z' \in \Z^d} c_{z'} \alpha_{j,z - 2z'}$ (and similar for the $\beta$'s). This {\it dilation equation} is often used for motivating and justifying thresholding in the conventional wavelet setting.

\ppn Now, substituting in (\ref{eqn:coefficients1}) yields
\begin{align*}
\hat{\alpha }_{j,z} &=  \frac{\Gamma(k)}{\Gamma(k+1/2)}\,\frac{1}{\sqrt{n}}\sum _{i=1}^n \varphi _{j,z}\left(X_i\right) \sqrt{V_{(k);i}} \\
&=  \frac{\Gamma(k)}{\Gamma(k+1/2)}\,\frac{1}{\sqrt{n}}\sum _{i=1}^n \left\lbrace \sum_{z' \in \Z^d} c_{z'} \varphi_{j+1,z' - 2 z}\left(X_i\right)  \right\rbrace \sqrt{V_{(k);i}} \\
&= \sum_{z' \in \Z^d} c_{z'} \left\lbrace \frac{\Gamma(k)}{\Gamma(k+1/2)}\,\frac{1}{\sqrt{n}} \,   \sum _{i=1}^n \varphi_{j+1,z' - 2 z}\left(X_i\right)  \sqrt{V_{(k);i}} \right\rbrace \\
&= \sum_{z' \in \Z^d} c_{z'} \hat{\alpha }_{j+1,z' - 2z},
\end{align*}
and similar for the $\hat{\beta}^{(q)}_{j,z}$'s from (\ref{eqn:coefficients2}). Hence, although the wavelet estimator developed in this paper is different in nature, the dilation equation applies to the estimated coefficients as it does in the conventional case. This suggests to carry on with thresholding for the shape-preserving estimator as well. 

\ppn Some numerical experiments were carried out and, indeed, it was seen that improvements could be obtained. Figure \ref{fig:threshold-plain} shows the shape-preserving estimator without thresholding on a typical sample of size $n=256$ from the Gaussian mixture (a) (see Section \ref{sec:numE}) using Daubechies wavelets with 6 vanishing moments, $k=1$, $j_0=0$ and $J=3$. This resolution is of course too high at this sample size (see Table \ref{tab:mise}), and the estimate is highly undersmoothed. Then soft thresholding was applied in (\ref{eqn:hatgJ}) on those estimated coefficients $\betahatqjz$ such that $|\beta_{j,k}| < C \sqrt{j + 1}/\sqrt{n}$, for an appropriate $C$ \citep{Delyon96}. The improvement is visually obvious (Figure \ref{fig:threshold-on}). The formal theoretical study of such a thresholding scheme is beyond the scope of this paper, though, and will be investigated in a follow-up paper. 

\begin{figure}[t!]
	\begin{subfigure}[t]{0.33\textwidth}
		\includegraphics[width=\textwidth]{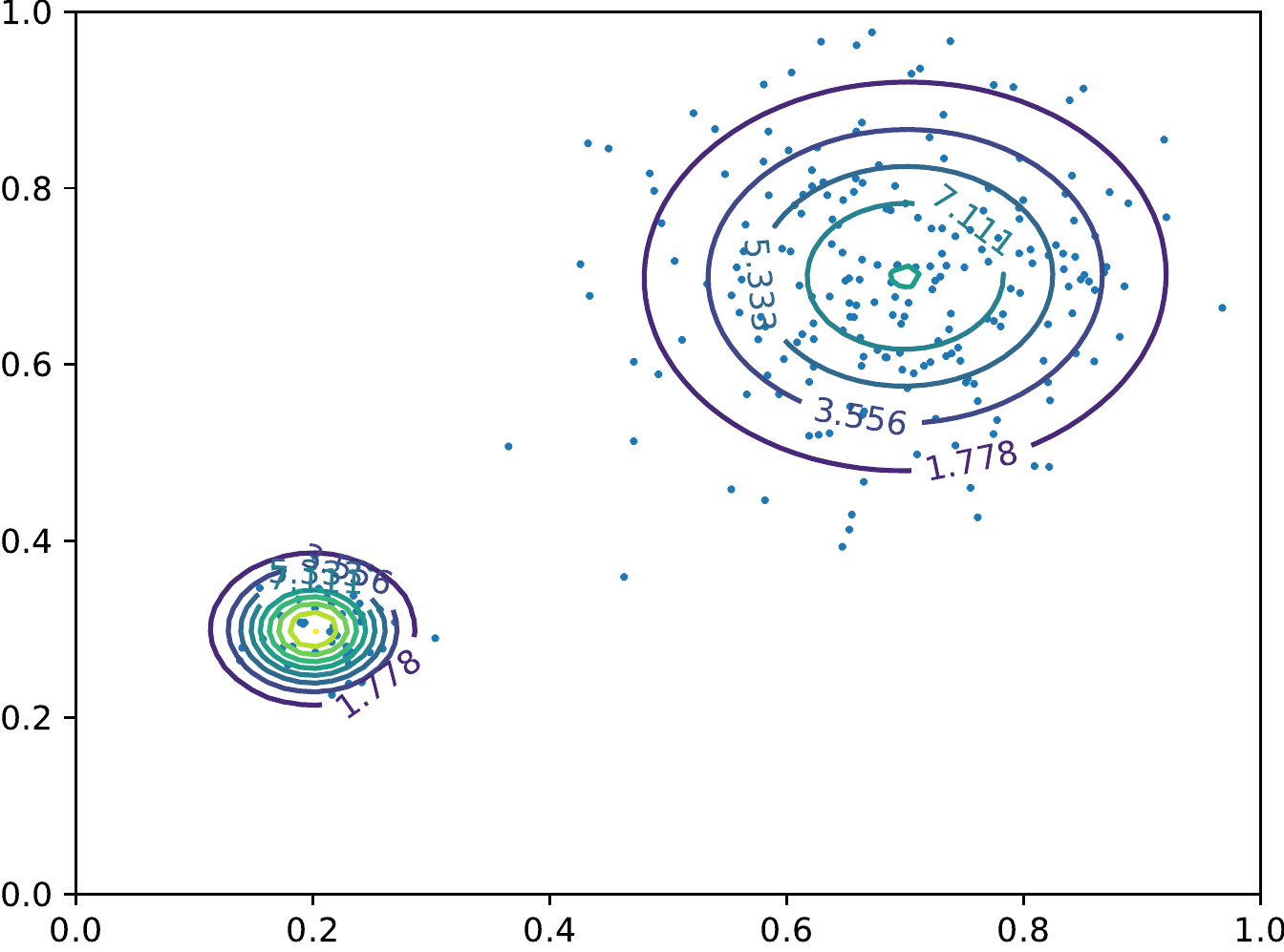}
		\caption{True density}
		\label{fig:threshold-true}
	\end{subfigure}
	\begin{subfigure}[t]{0.33\textwidth}
		\includegraphics[width=\textwidth]{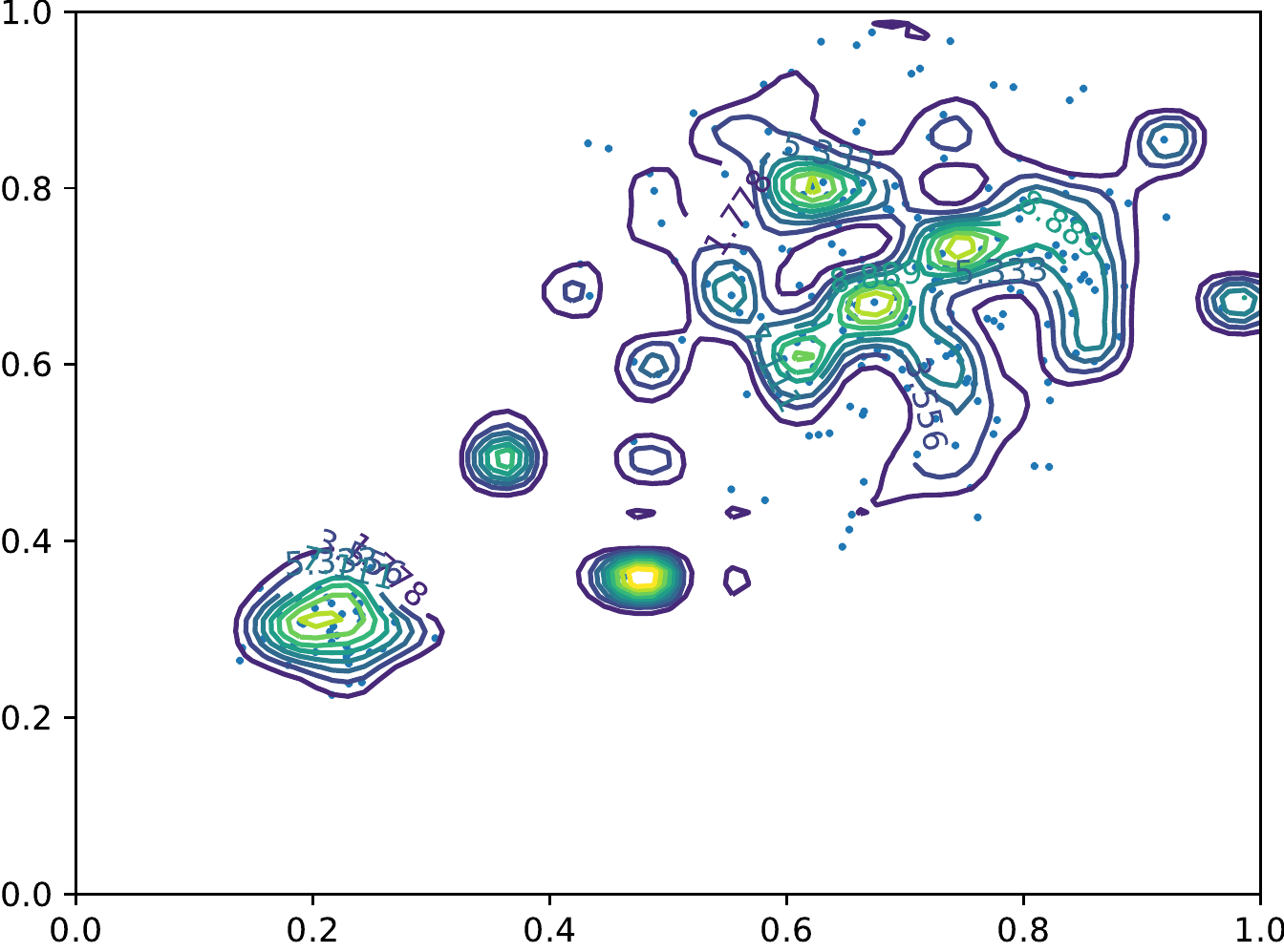}
		\caption{No thresholding}
		\label{fig:threshold-plain}
	\end{subfigure}
	\begin{subfigure}[t]{0.33\textwidth}
		\includegraphics[width=\textwidth]{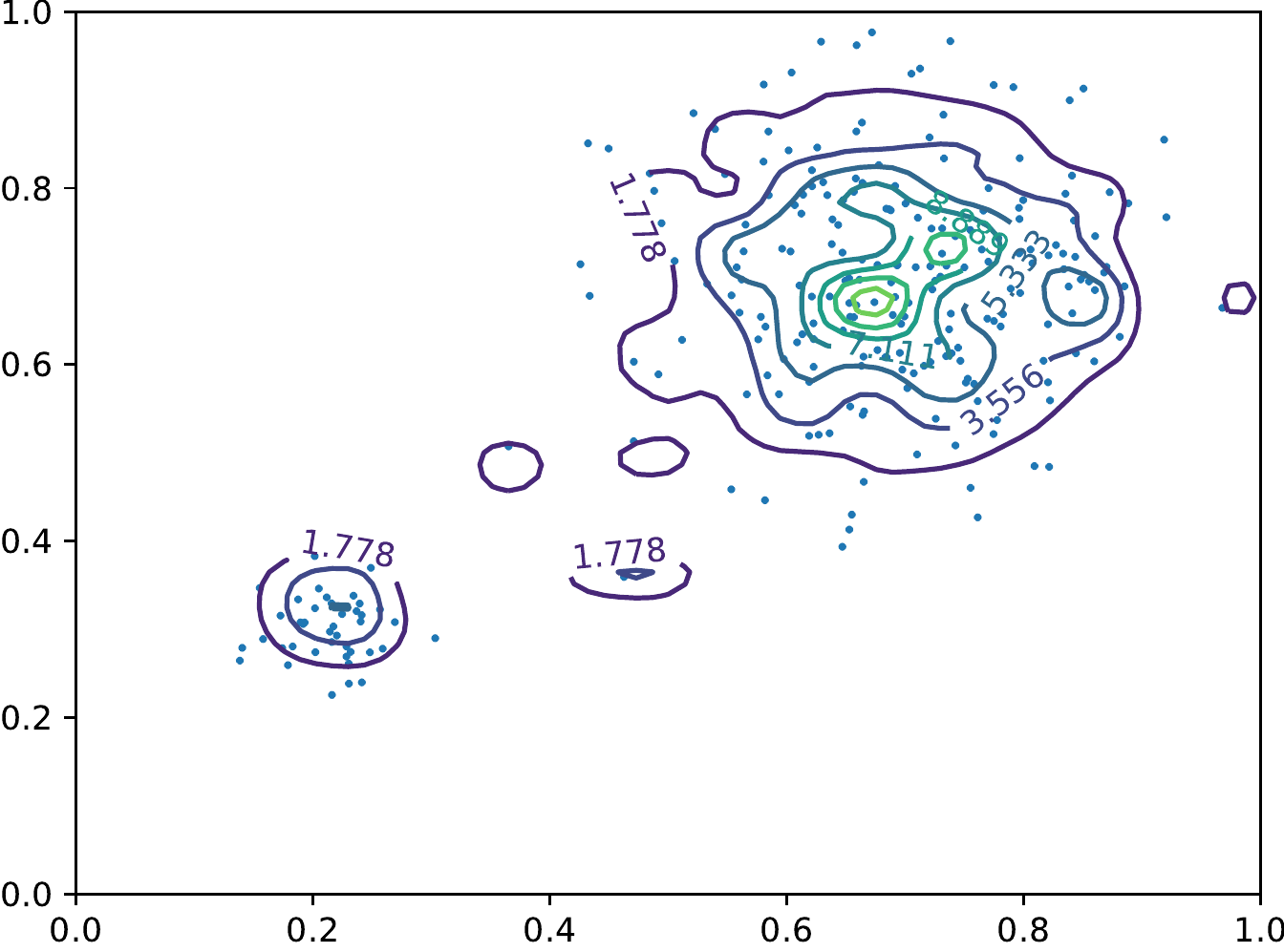}
		\caption{Thresholding}
		\label{fig:threshold-on}
	\end{subfigure}
	\caption{Shape-preserving estimates without (b) and with (c) thresholding for Gaussian mixture (a), $n=256, k=1, j_0=0, J=3$.}
	\label{fig:threshold}
\end{figure}

\ppn Our theoretical results also provide some avenue for dealing with generalisations of sums like (\ref{eqn:coefficients1}). For instance, \cite{Agarwal17} consider estimating the Fourier transform of the square-root of a probability density, viz.\ $\mathcal{F}\{ \sqrt{f}\}(\omega) = \int_{\R^d} e^{-i\omega x} \sqrt{f}(x)\diff x$, where $i = \sqrt{-1}$. Lemma \ref{lem:mainlemma} ensures that 
\[\frac{\Gamma(k)}{\Gamma(k+1/2)}\,\frac{1}{\sqrt{n}} \sum _{j=1}^n e^{- i \omega X_j} \sqrt{V_{(k);j}} \]
is an asymptotically unbiased estimate of $\mathcal{F}\{ \sqrt{f}\}(\omega)$ for all $\omega$'s, which could be used in several frameworks. Finally, the estimator proposed here may provide interesting benefits in more applied settings as well, for instance for image and shape recognition, in the spirit of \cite{Peter08} and \cite{Peter17}.  

\section*{Acknowledgements}

This research involved computations using the Linux computational cluster Katana supported by the Faculty of Science, UNSW Sydney. The work of Gery Geenens was supported by a Faculty Research Grant from the Faculty of Science, UNSW Sydney. The work of Spiridon Penev was partially supported by the Australian Government through the Australian Research Council's Discovery Projects funding scheme (project DP160103489).

\appendix
\section{Appendix} \label{sec:app}

\subsection*{Preliminaries}

First some preliminary concepts and technical results are presented.

\ppn For any convex and compact $C \subset \mathbb{R}^d$, let $\partial C$ denote its boundary. For $\eta > 0$, define the $\eta$-belt of $C$ as
\[C^{\left (<\eta \right )} = \left \{x  \in  C  \colon  \inf_{y\in \partial C}  \Norm{y-x} <\eta \right \},\]
the set of points in $C$ within Euclidean distance $\eta$ or less from $\partial C$. Also, we call $C^{\left (>\eta \right )}$ = $\mathit{C} \backslash C^{\left (<\eta \right )}$ the $\eta$-interior of $C$.
\ppn Fix $x\in C$, call $B_{x}\left(r\right)$ the ball of radius $r$ centred at $x$ and $\mu(B_x(r)) = c_0 r^d$ its volume ($\mu$ is the Lebesgue measure on $\R^d$, $c_0 = \frac{\pi^{d/2}}{\Gamma(d/2+1)}$). 
Results in \cite{Percus98} and \cite{Evans02} show that the following two properties hold for any compact and convex set $C \subset \R^d$:

\textbf{C1}. There exists $c_{2}>0$, independent of $x \in C$, such that for $r < \sup_{x,y \in C} \|x-y\|$, $\mu \left (B_{x}\left (r\right )\cap C\right ) \geq c_{2}  r^{d}$ ;

\textbf{C2}. There exist constants $\lambda  > 0$ and $c_{3}>0$
such that for all $0 < \eta  < \lambda$, $\mu \left(C^{\left(<\eta\right)}\right) < c_{3}\eta$.

\ppn The following technical lemma will be used repeatedly in the proofs below.

\begin{lemma}\label{lem:mainlemma} Let $\Xs = \{X_1,\ldots,X_n\}$ be a random sample from a distribution $F$ admitting a density $f$ supported on $C \subset \R^d$ satisfying Assumption \ref{ass:f}. Let $R_{(k);i}$ be the distance between $X_i$ and its $k$th nearest neighbour in the sample, as defined in Section \ref{subsec:motiv}. Let $\phi: \R^d \to \R$ be bounded on $C$ and $a > 0$ such that $\int_{\R^d} \phi\left(x\right)f \left(x\right)^{1-a}dx < \infty$.  Then, $\forall i \in \{1,\ldots,n\}$, as $n \to \infty$,
\begin{equation*}
\E\left(\phi\left(X_i\right)  R_{(k);i}^{a d} \right)  = \E\left(\phi\left(X_i\right) \E\left( R_{(k);i}^{a d}\big| X_i\right) \right) = \frac{1}{n^a} \frac{\Gamma \left(k + a\right)}{\Gamma\left(k\right)}\frac{1}{c_0^a} \left( \int_{\R^d} \phi\left (x\right )  f\left (x\right )^{1-a}  \diff x +O\left(n^{-1/d}\right)\right).  \end{equation*}
\end{lemma}

\begin{proof}
Call
\[\omega _{x}\left (r\right ) = \int _{B_{x}\left (r\right )}f  \left (z\right ) \,\diff z,  \]
the probability that the random variable $X \sim F$ falls in $B_{x}\left(r\right)$, and set $\omega_i(r) \doteq \omega _{X_i}(r)$ when referring to the ball centred at one particular observation $X_i$ from the sample. 
Let $F_{(k);i}$ be the distribution function of $R_{(k);i}$ for fixed $X_i$, that is, $F_{(k);i}(r) = \P(R_{(k);i} \leq r|X_i)$. With $X_i$ fixed, Lemma 4.1 in \cite{Evans02} writes
\[\diff F_{(k);i}(r) = k \binom{n-1}{k} \omega_i(r)^{k-1} (1-\omega_i(r))^{n-k-1} \,\diff\omega_i(r). \]
Hence
\[\E\left( R_{(k);i}^{ad}\big|X_i\right) =  k \binom{n-1}{k} \int_0^{c_1} r^{ad} \omega_i(r)^{k-1} (1-\omega_i(r))^{n-k-1} \,\diff\omega_i(r).\]
Since $f$ is positive on $C$ and $C$ is convex, $\omega_i(r)$ is strictly increasing for $0\leq r  \leq r_0$ for some $r_0$, and $\omega_i(r) \equiv 1$ for $r_0 \leq r$. Writing $h_i(\omega)$ for the inverse function $\omega_i^{-1}$ (where it exists), a change of variable yields
\[\E\left( R_{(k);i}^{ad}\big|X_i\right) =  k \binom{n-1}{k} \int_0^{1} h_i(\omega)^{ad} \omega^{k-1} (1-\omega)^{n-k-1} \diff\omega. \]
Define $\delta_n \doteq n^{-1/d}$, and break this expectation down into
\begin{align*} \E\left( R_{(k);i}^{ad}\big|X_i\right) & =  k \binom{n-1}{k} \int_0^{\omega_i(\delta_n)} h_i(\omega)^{ad} \omega^{k-1} (1-\omega)^{n-k-1} \diff\omega + k \binom{n-1}{k} \int_{\omega_i(\delta_n)}^1 h_i(\omega)^{ad} \omega^{k-1} (1-\omega)^{n-k-1} \diff\omega \\
& = k \binom{n-1}{k} \int_0^{\omega_i(\delta_n)} h_i(\omega)^{ad} \omega^{k-1} (1-\omega)^{n-k-1} \diff\omega + O(n^{-b})
\end{align*}
for all $b>0$, uniformly in $X_i$, as per Lemma 5.3 of \cite{Evans02}.  

\ppn Now, with $h_x = \omega^{-1}_x$, see that
\begin{align} \E\left(\phi\left(X_i\right) \E\left( R_{(k);i}^{a d}\big| X_i\right) \right) & = \int_C \phi(x) \left\{  k \binom{n-1}{k} \int_0^{1} h_x(\omega)^{ad} \omega^{k-1} (1-\omega)^{n-k-1} \diff\omega\right\}f(x)\diff x \label{eqn:EphiER} \\ &  = \int_C \phi(x) \left\{  k \binom{n-1}{k} \int_0^{\omega_x(\delta_n)} h_x(\omega)^{ad} \omega^{k-1} (1-\omega)^{n-k-1} \diff\omega\right\}f(x)\diff x + O(n^{-b}), \notag
\end{align}
as $\phi$ and $f$ are bounded on the compact $C$. As $b$ can be taken arbitrarily large, the remainder term can be neglected in front of any term tending to 0 polynomially fast. Hence, (asymptotically) all contribution to the inner integral in (\ref{eqn:EphiER}) comes from the set $\omega \in (0,\omega_x(\delta_n))$, that is, when $R_{(k);i}$ is smaller than $\delta_n$.

\ppn Now, write (\ref{eqn:EphiER}) as
\[\int_C \ldots \diff x = \int_{C^{(>\delta_n)}} \ldots \diff x + \int_{C^{(<\delta_n)}} \ldots \diff x \doteq (I) + (II)\]
with $C^{(>\delta_n)}$ and  $C^{(<\delta_n)}$ the $\delta_n$-interior and $\delta_n$-belt of $C$ as defined above.

\ppn {\bf Integral $(I)$:} $\int_{C^{(>\delta_n)}} \ldots \diff x$, hence $x \in$ $\delta_n$-interior and the distance from $x$ to $\partial C$ is at least $\delta_n$. Hence for all $r \leq \delta_n$, $B_x(r) \cap C = B_x(r)$. The first mean value theorem for definite integrals establishes the existence of $\xi _{1} \in B_{x}\left (r\right) \subset C$ such that
\begin{equation} \omega _{x}\left (r\right ) = \int _{B_{x}\left (r\right )}f  \left (z\right )  \diff z = f  \left (\xi _{1}\right )  \mu   \left (B_{x}\left (r\right )\right) = f(\xi_1)c_0 r ^d. \label{eqn:omegaxr} \end{equation}
By the mean value theorem, there is a $\xi _{2}$ between $x$ and $\xi_1$, hence $\xi_2 \in B_{x}\left (r\right ) \subset C$, such that $f  \left (\xi _{1}\right ) = f\left(x\right )+\nabla f  \left (\xi _{2}\right )' \left (x-\xi_{1}\right)$. Because $\xi_1 \in B_x(\delta_n)$ and $\|\nabla f  \left (\xi_{2}\right) \| < M$ for an absolute constant $M$ (the partial derivatives of $f$ are uniformly bounded on $C$ by Assumption \ref{ass:f}), we have $\left|f\left(\xi_{1}\right )-f\left (x\right) \right| < \delta_n M$ and hence $f\left (\xi _{1}\right ) = f  \left (x \right ) + O\left(\delta_n\right)$. Substitution in (\ref{eqn:omegaxr}) gives $\omega _{x}\left (r\right ) = \left (f  \left (x \right ) + O  \left (\delta_n \right )\right )c_0 r ^d$. As $f$ is bounded from below, this means that, as $n \to \infty$,
\[h_x(\omega)=\left(\frac{\omega}{c_0f(x)} \right)^{1/d} (1+O(\delta_n)), \]
where the $O(\delta_n)$-term holds uniformly in $x$ and $\omega$. This can be substituted in the inner integral of (\ref{eqn:EphiER}), and we obtain
\begin{align*}
\int_{C^{(>\delta_n)}} \phi(x) & \left\{  k \binom{n-1}{k} \int_0^{1} h_x(\omega)^{ad} \omega^{k-1} (1-\omega)^{n-k-1} \diff\omega\right\}f(x) \diff x \\
& = \frac{\Gamma(n)}{\Gamma(k)\Gamma(n-k)}(1+O(\delta_n)) \int_{C^{(>\delta_n)}}  \frac{\phi(x)f^{1-a}(x)}{c^a_0} \diff x \int_0^1 \omega^{a+k-1}(1-\omega)^{n-k-1}\diff \omega  \\
& = \frac{\Gamma(n)}{\Gamma(k)\Gamma(n-k)} \frac{\Gamma(k+a)\Gamma(n-k)}{\Gamma(n+a)}(1+O(\delta_n)) \int_{C^{(>\delta_n)}}  \frac{\phi(x)f^{1-a}(x)}{c^a_0}\diff x  \\
& = \frac{\Gamma(n)}{\Gamma(n+a)} \frac{\Gamma(k+a)}{\Gamma(k)}\frac{1}{c_0^a}(1+O(\delta_n)) \int_{C^{(>\delta_n)}}  \phi(x)f^{1-a}(x)\diff x.
\end{align*}

\ppn Now, given that $f$ is bounded from below and above on $C$, $f(x)^{1-a} \leq a_3 \doteq \max\{(1/a_1)^{1-a},a_2^{1-a}\}$, and by {\bf C2} above, $\mu(C^{(<\delta_n)})< c_3 \delta_n$ for $n$ large enough. So,
\[ \left|\int_{C^{(<\delta_n)}} \phi(x)f^{1-a}(x)\diff x\right| \leq \sup_{x \in C^{(<\delta_n)}} |\phi(x)| a_3 c_3 \delta_n = O(\delta_n),\]
as $n \to \infty$. Therefore,
\[\int_{C^{(>\delta_n)}}  \phi(x)f^{1-a}(x)\diff x= \int_C \phi(x)f^{1-a}(x)\diff x +O(\delta_n) = \int_{\R^d} \phi(x)f^{1-a}(x)\diff x +O(\delta_n).\]
Noting that $\Gamma(n)/\Gamma(n+a) = n^{-a}(1+O(n^{-1})) = n^{-a}(1+O(\delta_n))$, we finally get
\begin{multline} \int_{C^{(>\delta_n)}} \phi(x)  \left\{  k \binom{n-1}{k} \int_0^{1} h_x(\omega)^{ad} \omega^{k-1} (1-\omega)^{n-k-1} \diff\omega\right\}f(x) \diff x \\ = \frac{1}{n^a} \frac{\Gamma \left(k + a\right)}{\Gamma\left(k\right)} \frac{1}{c_0^a}\left( \int_{\R^d}\phi\left (x\right )  f\left (x\right )^{1-a}  \diff x +O\left(n^{-1/d}\right)\right). \label{eqn:I} \end{multline}

\ppn {\bf Integral $(II)$:} $\int_{C^{(<\delta_n)}} \ldots \diff x$, hence we can no more assume that $B_{x}(r)\subset C$. However, as $\sup_{x\in C^{(<\delta_n)}} f(x) \leq \sup_{x\in C} f(x) \leq a_2$ and $\mu(B_x(r) \cap C) < \mu(B_x(r) = c_0 r^d$, it holds $\omega_x(r) < a_2 c_0 r^d$. An upper bound for its inverse is thus $h_x(\omega) \leq (a_2 c_0)^{-1/d} \omega^{1/d}$. Hence,
\begin{align*}
|(II)| & \leq \int_{C^{(<\delta_n)}} |\phi(x)| \left\{  k \binom{n-1}{k} \int_0^{1} h_x(\omega)^{ad} \omega^{k-1} (1-\omega)^{n-k-1} \diff\omega\right\}f(x)\diff x \\
& \leq \int_{C^{(<\delta_n)}} |\phi(x)| \left\{  k \binom{n-1}{k} \int_0^{1} (a_2 c_0)^{-a} \omega^{a+k-1} (1-\omega)^{n-k-1} \diff\omega\right\}f(x)\diff x \\
& = \frac{\Gamma(a+k)}{\Gamma(k)} \,\frac{\Gamma(n)}{\Gamma(n+a)} (a_2 c_0)^{-a} \int_{C^{(<\delta_n)}} |\phi(x)|f(x)\diff x \\
& \leq \frac{\Gamma(a+k)}{\Gamma(k)} \,\frac{\Gamma(n)}{\Gamma(n+a)} (a_2 c_0)^{-a} \sup_{x \in C} |\phi(x)| a_2 c_3 \delta_n,
\end{align*}
by {\bf C2} above. Thus,
\begin{equation} |(II)| \leq \frac{\Gamma(k+a)}{\Gamma(k)} O(n^{-a} \delta_n) = \frac{\Gamma(k+a)}{\Gamma(k)} O(n^{-a-1/d} ), \qquad \text{ as } n \to \infty. \label{eqn:II} \end{equation}

\ppn Putting together (\ref{eqn:I}) and (\ref{eqn:II}) in (\ref{eqn:EphiER}), it follows
\[  \E\left(\phi\left(X_i\right) R_{(k);i}^{a d} \right) = \frac{1}{n^a} \frac{\Gamma \left(k + a\right)}{\Gamma\left(k\right)}\frac{1}{c_0^a} \left( \int_{\R^d}\phi\left (x\right )  f\left (x\right )^{1-a}  \diff x +O\left(n^{-1/d}\right)\right),\]
as announced.
\end{proof}

\subsection*{Proof of Proposition \ref{prop:coefconsist}}

The proof is given for the coefficients $\hat{\alpha}_{j,z}$. The proof for the coefficients $\hat{\beta}^{(q)}_{j,z}$ is identical.

\ppn {\bf Bias:} From (\ref{eqn:coefficients1}), we have with (\ref{eqn:vol})
\[ \E(\hat{\alpha}_{j,z} ) = \E\left(\frac{\Gamma(k)}{\Gamma(k+1/2)}\,\frac{1}{\sqrt{n}}\sum _{i=1}^n \varphi _{j,z}\left(X_i\right) \sqrt{V_{(k);i}} \right) = n^{1/2}\,\frac{\Gamma(k)}{\Gamma(k+1/2)} \sqrt{c_0}\,\E\left( \varphi _{j,z}\left(X_1\right) R^{d/2}_{(k);1}\right).\]
Applying Lemma \ref{lem:mainlemma} with $\phi = \varphi_{j,z}$ and $a = 1/2$ yields
\[\E\left( \varphi _{j,z}\left(X_1\right) R^{d/2}_{(k);1}\right) = n^{-1/2}\,\frac{\Gamma(k+1/2)}{\Gamma(k)}\,\frac{1}{\sqrt{c_0}}\left(\int_{\R^d}\varphi _{j,z}(x) \sqrt{f}(x)\diff x + O(n^{-1/d})\right),\]
which gives
\[ \E(\hat{\alpha}_{j,z} ) =  \int_{\R^d}\varphi _{j,z}(x) \sqrt{f}(x)\diff x + O(n^{-1/d}) = \alpha_{j,z} + O(n^{-1/d}). \]

\ppn {\bf Variance:} Lemma 4.6$(ii)$ of \cite{Evans08} gives an upper bound on the variance of statistics of type $S_n = \sum_{i=1}^n h_{i,n}(\Xs)$, where $h_{i,n}(\Xs)$ is an arbitrary (measurable) function of the sample point $X_i$ and its $k$-nearest neighbours among the sample $\Xs$. Take here
\[h_{i,n}(\Xs) \doteq  \varphi _{j,z}\left(X_i\right) \sqrt{V_{(k);i}}\]
and see that $\hat{\alpha}_{j,z} = \frac{\Gamma(k)}{\Gamma(k+1/2)}\,\frac{1}{\sqrt{n}} S_n$. Lemma 4.6$(ii)$ of \cite{Evans08} reads
\begin{equation} \var(S_n) \leq 2(n+1) (3+8k^2 dc_0) \E\left(h^2_{i,n}(\Xs)\right), \label{eqn:varSn} \end{equation}
for $n \geq 16k$. Here,
\begin{align*}
\E\left(h^2_{i,n}(\Xs)\right) & = \E\left(\varphi^2_{j,z}\left(X_i\right) V_{(k);i} \right) \\
& = c_0 \E\left(\varphi^2_{j,z}\left(X_i\right) R^d_{(k);i} \right) \\
& = \frac{k}{n} \left( \int_{\R^d}\varphi^2_{j,z}(x)\diff x + O(n^{-1/d})\right),
\end{align*}
from Lemma \ref{lem:mainlemma} with $\phi = \varphi^2_{j,z}$ and $a=1$. By definition, $\int_{\R^d}\varphi^2_{j,z}(x)\diff x =1 $ (orthonormal wavelet basis, Assumption \ref{ass:wav}), hence $\E\left(h^2_{i,n}(\Xs)\right) = \frac{k}{n} (1+ O(n^{-1/d}))$. From this and (\ref{eqn:varSn}), we obtain
\[ \var(\hat{\alpha}_{j,z}) \leq \left(\frac{\Gamma(k)}{\Gamma(k+1/2)}\right)^2\,\frac{1}{n} \,2(n+1) (3+8k^2 dc_0) \frac{k}{n} (1+ O(n^{-1/d})) =  k^3\left(\frac{\Gamma(k)}{\Gamma(k+1/2)}\right)^2 O(n^{-1}), \]
as $n \to \infty$. \qed

\subsection*{Proof of Proposition \ref{prop:gJbias}}

From (\ref{eqn:hatgJ2}) we have
\begin{align}
\hat{g}_J(x) & = \sum_{z \in \Z^d} \hat{\alpha}_{J+1,z} \varphi_{J+1,z}(x) \notag \\
& = \sum_{z \in \Z^d} \frac{\Gamma(k)}{\Gamma(k+1/2)}\,\frac{1}{\sqrt{n}}\sum _{i=1}^n \varphi _{J+1,z}\left(X_i\right) \sqrt{V_{(k);i}} \,\varphi_{J+1,z}(x) \notag  \\
& = \frac{\Gamma(k)}{\Gamma(k+1/2)}\,\frac{\sqrt{c_0}}{\sqrt{n}}\,\sum_{i=1}^n R_{(k);i}^{d/2} \sum_{z \in \Z^d} \varphi_{J+1,z}\left(X_i\right) \varphi_{J+1,z}(x) \notag  \\
& = \frac{\Gamma(k)}{\Gamma(k+1/2)}\,\frac{\sqrt{c_0}}{\sqrt{n}}\,\sum_{i=1}^n R_{(k);i}^{d/2} K_{J+1}(x,X_i),  \label{eqn:gJK}
\end{align}
hence
\[ \E\left(\hat{g}_J(x) \right)  = \frac{\Gamma(k)}{\Gamma(k+1/2)}\,\sqrt{n}\,\sqrt{c_0}\, \E\left(K_{J+1}(x,X_1) R_{(k);1}^{d/2} \right). \]
Lemma \ref{lem:mainlemma} with $\phi = K_{J+1}(x,\cdot)$ and $a =1/2$ establishes that
\[\E\left(K_{J+1}(x,X_1) R_{(k);1}^{d/2} \right) = \frac{1}{\sqrt{n}}\,\frac{\Gamma(k+1/2)}{\Gamma(k)}\,\frac{1}{\sqrt{c_0}} \left( \int_{\R^d}K_{J+1}(x,y) \sqrt{f}(y)\diff y + O(n^{-1/d})\right), \]
and inspection of the proof of Lemma \ref{lem:mainlemma} reveals that the $O(n^{-1/d})$ term holds uniformly in $x \in C$. This means that
\[\E\left(\hat{g}_J(x) \right) = \int_{\R^d}K_{J+1}(x,y) \sqrt{f}(y)\diff y + O(n^{-1/d}) = K_{J+1}\sqrt{f}(x) + O(n^{-1/d}), \]
as $n \to \infty$, uniformly in $x \in C$, proving $(i)$.

\ppn It follows from (\ref{eqn:gJK}) as well that
\[\var\left(\frac{\Gamma(k+1/2)}{\Gamma(k)}\,\sqrt{\frac{n}{k^{3}}}\,\hat{g}_J(x) \right) = \frac{c_0}{k^3} \var\left(\sum_{i=1}^n h_{i,n}(\Xs)\right) \]
where here $h_{i,n}(\Xs) \doteq K_{J+1}(x,X_i) R^{d/2}_{(k);i}$. Lemma \ref{lem:mainlemma} with $a=1$ and $\phi = K^2_{J+1}(x,\cdot)$ yields
\[ \E\left(h^2_{i,n}(\Xs) \right) = \frac{k}{c_0n}\left(\int_{\R^d} K^2_{J+1}(x,y)\diff y +O(n^{-1/d})\right) \]
(with again the $O(n^{-1/d})$-term holding uniformly in $x \in C$). Hence, for $n \geq 16k$, Lemma 4.6$(ii)$ of \cite{Evans08} gives
\[\var\left(\sum_{i=1}^n h_{i,n}(\Xs)\right) \leq 2(n+1)(3+8k^2dc_0) \frac{k}{c_0n}\left(\int_{\R^d} K^2_{J+1}(x,y)\diff y +O(n^{-1/d})\right),\]
whereby
\[\var\left(\frac{\Gamma(k+1/2)}{\Gamma(k)}\,\sqrt{\frac{n}{k^{3}}}\,\hat{g}_J(x) \right) \leq \text{constant} \times \int_{\R^d} K^2_{J+1}(x,y)\diff y +O(n^{-1/d}).\]
This establishes $(ii)$. \qed

\subsection*{Proof of Theorem \ref{thm:uniconst}}

The Mean Integrated Squared Error (MISE) $\E\left(\|\hat{g}_J - \sqrt{f} \|_2^2\right)$ can classically be decomposed into the integrated squared bias and the integrated variance:
\begin{equation} \E\left(\|\hat{g}_J - \sqrt{f} \|_2^2\right) = \|\E\left(\hat{g}_J\right) - \sqrt{f} \|_2^2 + \E\left(\|\hat{g}_J - \E\left(\hat{g}_J\right) \|_2^2\right). \label{eqn:MISE}\end{equation}
For the bias term, it follows from Proposition \ref{prop:gJbias}$(i)$ that
\[\|\E\left(\hat{g}_J\right) - \sqrt{f} \|_2 \leq \|K_{J+1}\sqrt{f} - \sqrt{f} \|_2 + O(n^{-1/d}).\]
As $f \in B^{m,2}(L)$ implies $\sqrt{f} \in B^{m,2}(L')$ for some $0 \leq L' <\infty$, one can call (multivariate versions of) Theorem 8.1$(ii)$ and Corollary 10.1 of \cite{Hardle98} to obtain
\[\sup_{f \in B^{m,2}(L)}\|K_{J+1}\sqrt{f} - \sqrt{f} \|_2 \leq \kappa_1 2^{-J\,m},\]
for some constant $\kappa_1$. Hence, for $n$ large enough,
\begin{equation} \sup_{f \in B^{m,2}(L)}\|\E\left(\hat{g}_J\right) - \sqrt{f} \|_2 \leq \kappa_1 2^{-J\,m} + \kappa_2 n^{-1/d}, \label{eqn:bias}\end{equation}
for constants $\kappa_1,\kappa_2 < \infty$.

\ppn To evaluate  $\int_{\R^d} K^2_{J+1}(x,y)\diff y$ in the righ-hand side of Proposition \ref{prop:gJbias}$(ii)$ we use that
\begin{align*}
\int_{\R^d} K^2_{J+1}(x,y)\diff y & = \int_{\R^d} 2^{2d(J+1)} K^2(2^{J+1}x,2^{J+1}y)\diff y \\
& \leq 2^{2d(J+1)} \int_{\R^d} F^2(2^{J+1}(x-y))\diff y \\
& = 2^{(d+1)J} \int_{\R^d} F^2(v)\diff v,
\end{align*}
where Assumption \ref{ass:kern} justifies the inequality. It follows
\[\var\left(\hat{g}_{J}(x) \right) \leq \text{constant} \times n^{-1} k^3 \left(\frac{\Gamma(k)}{\Gamma(k+1/2)} \right)^2 \left(2^{dJ} \int_{\R^d} F^2(v)\diff v + O(n^{-1/d})\right),\]
which can be integrated over the compact $C$:
\[\E\left(\|\hat{g}_J - \E\left(\hat{g}_J\right) \|_2^2 \right)= \int_{\R^d} \var\left(\hat{g}_J(x) \right) \diff x \leq \text{constant} \times n^{-1} k^3 \left(\frac{\Gamma(k)}{\Gamma(k+1/2)} \right)^2 \left(2^{dJ} \int_{\R^d} F^2(v)\diff v + O(n^{-1/d})\right).\]
Hence, for $n$ large enough, there exists a constant $\kappa'_3<\infty$ such that 
\begin{equation} \E\left(\|\hat{g}_J - \E\left(\hat{g}_J\right) \|_2^2 \right) < \kappa_3' n^{-1} k^3 \left(\frac{\Gamma(k)}{\Gamma(k+1/2)} \right)^2 2^{dJ}. \label{eqn:var}\end{equation}
Plugging (\ref{eqn:bias}) and (\ref{eqn:var}) in (\ref{eqn:MISE}) yields the result. \qed

\end{document}